\documentclass[sigconf,nonacm]{acmart}

\usepackage[inline]{enumitem}
\usepackage[ruled,linesnumbered]{algorithm2e} 
\usepackage{framed}

\PassOptionsToPackage{hyphens}{url}
\hypersetup{
    colorlinks,
    linkcolor={T-Q-DT5},
    citecolor={T-Q-DT2},
    urlcolor={T-Q-DT1}
}

\usepackage[normalem]{ulem}

\usepackage{makecell} 
\usepackage{multirow} 
\usepackage{threeparttable}
\usepackage{tabularx}

\usepackage{xspace} 

\usepackage{colortbl}
\usepackage[OkabeIto, Tol]{colorblind}

\usepackage{pifont}

\usepackage{bm}

\renewcommand\footnotetextcopyrightpermission[1]{}

\usepackage{mathtools}

\usepackage{siunitx}
\usepackage{ragged2e}

\usepackage{pgfplots}
\pgfplotsset{compat=1.18}
\usepackage{pgfplotstable}

\newlength{\smalltabcolsep}

\newcommand{\pick}{\xleftarrow{\scriptscriptstyle \$}\xspace}
\newcommand{\qequal}{\stackrel{\scriptscriptstyle ?}=\xspace}
\newcommand{\bin}{\{0,1\}\xspace}

\newcommand{\lang}{\ensuremath{\mathcal{L}}\xspace}
\newcommand{\intervalvec}{\ensuremath{\vec{\Delta}}\xspace}
\newcommand{\messvec}{\ensuremath{\vv{m}}\xspace}
\newcommand{\ex}{\ensuremath{a}\xspace}
\newcommand{\rb}{\ensuremath{b}\xspace}
\newcommand{\ver}{\ensuremath{v}\xspace}
\newcommand{\vervec}{\ensuremath{\vec{\ver}}\xspace}
\newcommand{\paid}{\ensuremath{u}\xspace}
\newcommand{\paidvec}{\ensuremath{\vec{\paid}}\xspace}
\newcommand{\puzzvec}{\ensuremath{\vec{p}}\xspace} 
\newcommand{\commvec}{\ensuremath{\vec{g}}\xspace} 
\newcommand{\solsvec}{\ensuremath{\vec{s}}\xspace} 
\newcommand{\timevec}{\ensuremath{\vec{t}}\xspace}
\newcommand{\keyvec}{\ensuremath{\vec{k}}\xspace}
\newcommand{\timeexpvec}{\ensuremath{\vv{T}}\xspace}
\newcommand{\randvec}{\ensuremath{\vec{r}}\xspace}
\newcommand{\witnvec}{\ensuremath{\vec{d}}\xspace}
\newcommand{\expvec}{\ensuremath{\vec{\ex}}\xspace}
\newcommand{\extravec}{\ensuremath{\vec{\Psi}}\xspace}
\newcommand{\coins}{\ensuremath{\textit{coins}}\xspace}
\newcommand{\adr}{\ensuremath{\textit{adr}}\xspace}
\newcommand{\coinsvec}{\ensuremath{\vv{\textit{co{\i}ns}}}\xspace}

\newcommand{\sk}{\ensuremath{\textit{sk}}\xspace}
\newcommand{\pk}{\ensuremath{\textit{pk}}\xspace}
\newcommand{\csk}{\ensuremath{\textit{csk}}\xspace}
\newcommand{\cpk}{\ensuremath{\textit{cpk}}\xspace}
\newcommand{\skegen}{\ensuremath{\mathsf{SKE.keyGen}}\xspace}
\newcommand{\skeenc}{\ensuremath{\mathsf{SKE.Enc}}\xspace}
\newcommand{\skedec}{\ensuremath{\mathsf{SKE.Dec}}\xspace}
\newcommand{\Com}{\ensuremath{\mathsf{Com}}\xspace}
\newcommand{\Ver}{\ensuremath{\mathsf{Ver}}\xspace}

\newcommand{\cl}{\ensuremath{\mathcal{C}}\xspace} 
\newcommand{\se}{\ensuremath{\mathcal{S}}\xspace} 
\newcommand{\ve}{\ensuremath{\mathcal{V}}\xspace} 
\newcommand{\tpc}{\ensuremath{\mathcal{H\cl}}\xspace} 
\newcommand{\tps}{\ensuremath{\mathcal{H\se}}\xspace} 
\newcommand{\adv}{\ensuremath{\mathcal{A}}\xspace} 
\newcommand{\scc}{\ensuremath{\mathcal{SC}}\xspace} 
\newcommand{\hash}{\ensuremath{\mathsf{G}}\xspace}
\newcommand{\parse}{\ensuremath{\mathsf{parse}}\xspace}
\newcommand{\poly}{\ensuremath{\textit{poly}}\xspace}
\newcommand{\negl}{\ensuremath{\mu(\lambda)}\xspace}
\newcommand{\ced}{\ensuremath{\mathsf{CEDG}}\xspace}

\newcommand{\empt}{\ensuremath{\epsilon}\xspace}

\newcommand{\setup}[1]{\ensuremath{\mathsf{Setup_{#1}}\xspace}}

\newcommand{\gen}[1]{\ensuremath{\mathsf{GenPuzzle_{#1}}\xspace}}
\newcommand{\genext}[1]{\ensuremath{\mathsf{GenPuzzleExt_{#1}}\xspace}}
\newcommand{\solve}[1]{\ensuremath{\mathsf{Solve_{#1}}\xspace}}
\newcommand{\prove}[1]{\ensuremath{\mathsf{Prove_{#1}}\xspace}}
\newcommand{\verify}[1]{\ensuremath{\mathsf{Verify_{#1}}\xspace}}
\newcommand{\register}[1]{\ensuremath{\mathsf{Register_{#1}}\xspace}}
\newcommand{\pay}[1]{\ensuremath{\mathsf{Pay_{#1}}\xspace}}
\newcommand{\delegate}[1]{\ensuremath{\mathsf{Delegate_{#1}}\xspace}}
\newcommand{\retrieve}[1]{\ensuremath{\mathsf{Retrieve_{#1}}\xspace}}
\newcommand{\eval}{\ensuremath{\mathsf{Eval}\xspace}}

\newcommand{\dgm}{D-TLP\xspace}
\newcommand{\dgc}{ED-TLP\xspace}
\newcommand{\mtdgm}{DTLP\xspace}
\newcommand{\gm}{GM-TLP\xspace}
\newcommand{\gmp}{GMTLP\xspace}
\newcommand{\gc}{GC-TLP\xspace}
\newcommand{\ctlp}{C-TLP\xspace}
\newcommand{\tlp}{TLP\xspace}

\newcommand{\cmark}{\textcolor{OI3}{\ding{51}\xspace}}
\newcommand{\xmark}{\textcolor{OI6}{\ding{55}\xspace}}

\newtheorem{assumption}{Assumption}
\newtheorem{theorem}{Theorem}
\newtheorem{lemma}{Lemma}

\theoremstyle{definition}
\newtheorem{definition}{Definition}
\newtheorem{claim}{Claim}

\makeatletter
\def\brokenuline{\bgroup
\UL@pixel.3\p@ \ULdepth .25ex
\let\UL@putbox\UL@brokenputbox \ULset}
\def\UL@brokenputbox{%
\ifx\UL@start%
\@empty%
\else 
\vrule\@width\z@ \LA@penalty\@M
{\UL@skip\wd\UL@box \advance\UL@skip\UL@pixel \UL@leaders \kern-\UL@skip}%
\textcolor{white}{%
\kern-2\UL@pixel\copy\UL@box\kern-\wd\UL@box
\kern \UL@pixel\raise.4\p@\copy\UL@box\kern-\wd\UL@box
\kern \UL@pixel\lower.4\p@\copy\UL@box\kern-\wd\UL@box
\kern \UL@pixel\raise.4\p@\copy\UL@box\kern-\wd\UL@box
\kern \UL@pixel\copy\UL@box\kern-\wd\UL@box
\kern-2\UL@pixel}%
\box\UL@box \fi}

\newcommand\blfootnote[1]{
    \begingroup
    \renewcommand\thefootnote{}\footnote{#1}
    \addtocounter{footnote}{-1}
    \endgroup
}

\begin{document}

\title{Scalable Time-Lock Puzzles}

\author{Aydin Abadi}
\affiliation{%
  \institution{Newcastle University}
  \city{Newcastle}
  \country{United Kingdom}}
  \email{aydin.abadi@ncl.ac.uk}

\author{Dan Ristea}
\affiliation{%
  \institution{University College London}
  \city{London}
  \country{United Kingdom}}
  \email{dan.ristea.19@ucl.ac.uk}

\author{Artem Grigor}
\affiliation{%
  \institution{University of Oxford}
  \city{Oxford}
  \country{United Kingdom}}
  \email{artem.grigor@cs.ox.ac.uk}

\author{Steven J.\ Murdoch}
\affiliation{%
  \institution{University College London}
  \city{London}
  \country{United Kingdom}
}
\email{s.murdoch@ucl.ac.uk}



\begin{abstract}
Time-Lock Puzzles (TLPs) enable a client to lock a message such that a server can unlock it only after a specified time. They have diverse applications, such as scheduled payments, secret sharing, and zero-knowledge proofs. 
In this work, we present a scalable TLP designed for real-world scenarios involving a large number of puzzles, where clients or servers may lack the computational resources to handle high workloads. Our contributions are both theoretical and practical. From a theoretical standpoint, we formally define the concept of a ``Delegated Time-Lock Puzzle (D-TLP)'', establish its fundamental properties, and introduce an \textit{upper bound} for TLPs, addressing a previously overlooked aspect. From a practical standpoint, we introduce the ``Efficient Delegated Time-Lock Puzzle'' (ED-TLP) protocol, which implements the D-TLP concept. This protocol enables both the client and server to \textit{securely outsource} their resource-intensive tasks to third-party helpers. It enables \textit{real-time verification} of solutions and guarantees their delivery within predefined time limits by integrating an upper bound and a fair payment algorithm. ED-TLP allows combining puzzles from different clients, enabling a solver to process them sequentially, significantly reducing computational resources, especially for a large number of puzzles or clients. ED-TLP is the first protocol of its kind. We have implemented ED-TLP and conducted a comprehensive analysis of its performance for up to 10,000 puzzles. The results highlight its significant efficiency in TLP applications, demonstrating that ED-TLP securely delegates 99\% of the client's workload and 100\% of the server's workload with minimal overhead. 
\end{abstract}
\maketitle

\blfootnote{A preliminary version of this paper appears in ACM AsiaCCS 2025.}

\section{Introduction}\label{sec:introduction}
Time-Lock Puzzles~(TLP) are elegant cryptographic primitives that enable the transmission of information to the future.
A \textit{client} can lock a message in a \tlp such that no \textit{server} can unlock it for a specified duration.
\citeauthor{TimothyMay1993}~\cite{TimothyMay1993} was first to propose sending information into the future, i.e., time-lock puzzle/encryption.
As May's scheme requires a trusted agent to release a secret on time, \citeauthor{Rivest:1996:TPT:888615}~\cite{Rivest:1996:TPT:888615} proposed a trustless TLP RSA-based using sequential modular squaring which became the basis of much of the work in the field.
It is secure against a server with high computation resources that support parallelization.
Applications of TLP include timely payments in cryptocurrencies~\cite{ThyagarajanBMDK20}, e-voting~\cite{ChenD12}, sealed-bid auctions~\cite{Rivest:1996:TPT:888615}, timed secret sharing~\cite{kavousi2023timed},
timed commitments \cite{KatzLX20},
zero-knowledge proofs \cite{dwork2000zaps},
and verifiable delay functions~\cite{BonehBBF18}.

To avoid requiring the client or another party trusted by the client to be online at the time of release, TLPs necessarily approximate time using computation,
which impose high costs, especially in settings where multiple puzzles are required.
In consequence, techniques to improve the scalability of TLPs have been proposed in the literature.
However, there are generic settings these techniques do not address.

Consider two clients that need to send unrelated messages to a
\textit{resource-constrained} server.
One client requires their message to be hidden for 24 hours; the other for 40 hours.
Both clients encode their messages as RSA-based time-lock puzzles and send them independently to the server. To obtain both messages on time, the server would have to solve the puzzles in parallel.
It would thus be performing total work equivalent to at least 64 hours of CPU time, an obvious inefficiency.

None of the techniques in the literature can reduce the workload in this setting.
Homomorphic TLPs~\cite{MalavoltaT19,BrakerskiDGM19,abadi2024temporafusiontimelockpuzzleefficient} enable computation on multiple puzzles to produce a single TLP containing the result and forgo solving individual puzzles separately.
Unfortunately, these cannot help in this setting, as the two puzzles are unrelated.
Batchable TLPs~\cite{SrinivasanLMNPT23}
cannot help either, as they only consider identical time intervals.
Chained~TLP~(C-TLP)~\cite{Abadi-C-TLP} enable a client to encode multiple puzzles separated by a fixed interval.
The server can obtain each message on time only solving each chained puzzle sequentially, rather than all puzzles in parallel.
However, C-TLP requires that \emph{a single client} generate all puzzles, so two clients cannot chain their puzzles, and the interval between puzzles is fixed, so arbitrary intervals cannot be chained (at least not without incurring a high communication and computational cost,
as discussed in \autoref{subsec::strawman}).
Therefore, none of the existing TLP variants enable the server to only do work equivalent to the longest puzzle while obtaining both messages at their intended times.

Moreover, for the fundamental feature of TLPs to hold -- that they cannot be solved before the intended time -- puzzles must be set assuming the capabilities of the most powerful solver in the world.
Hence, a less capable server  may still struggle to solve even a single, combined puzzle in time.
This also introduces issues in settings with multiple servers with heterogeneous capabilities, as the servers will solve the puzzles at different times.

Setting puzzles for the most capable solver introduces its own challenge.
Solving capabilities are difficult to estimate and, for longer-running puzzles, to predict how they will improve.
The~\textit{LCS35 puzzle} released by \citeauthor{LCS35}~\cite{LCS35} is a poignant example.
Intended to take 35 years, it was solved in 3.5 years on a general purpose CPU~\cite{LCS35Solve}.
Subsequently, it was opened in 2 months~\cite{Cryptophage,CryptophageSolver} using field-programmable gate arrays~(FPGA), showing that hardware specialized in repeated squaring provides a large speed-up.

Existing schemes consider the lower bound, representing the earliest possible time for the most powerful server to find a solution.
However, an \emph{upper bound} indicating the guaranteed time for a regular server to find a solution has not been studied and defined.

\subsection{Our Contributions}\label{subsec:contributions}
These limitations inform the key features of the protocols in this paper.
Critically, a solution requires secure delegation.
On the server-side, delegation helps a server without capabilities offload their work to more powerful helpers.
On the client-side, delegation helps clients combine unrelated puzzles into a single chain.
To support this, a solution must also allow for arbitrary intervals in the chain.

We therefore introduce the concept of the ``Delegated Time-Lock Puzzle'' (\dgm) and present the ``Efficient Delegated Time-Lock Puzzle'' (\dgc), a protocol that realizes this concept.

\dgc is a modular protocol that enables secure end-to-end delegation, efficiently handles the multi-puzzle setting, and accommodates varied time intervals.
It lets clients and servers delegate tasks to potentially adversarial third-party helpers, who may have specialized hardware.
This is especially beneficial when servers have varied capabilities or even lack the capabilities to directly engage in the protocol.
Although \dgc includes both client and server delegation, each can be used separately.
The protocol preserves the privacy of plaintext solutions and supports efficient real-time verification of solutions by relying solely on symmetric-key primitives.
In a multi-client setting, when combined with client delegation, this lets multiple clients combine distinct messages into a single puzzle chain.
The protocol supports resuming a previously generated puzzle chain, which allows a client to act as a helper for other clients by adding additional puzzles to an existing puzzle.

\dgc offers \textit{accurate time} solution delivery, ensuring that a rational helper always delivers a solution before a predefined time.
To capture this feature, we rely on two techniques.
Firstly, we equip \dgc with an explicit upper bound for a helper to deliver a solution, determined by a ``Customized Extra Delay Generating'' function, which may hold independent interest and find applications in other delay-based primitives, such as Verifiable Delay Functions~(VDFs)~\cite{BonehBBF18,Wesolowski19,EphraimFKP20a}.
Secondly, we incorporate fair payment into \dgc, an algorithm that only compensates a solver if it provides a valid solution before a predefined time, made trustless through a smart contract.

We formally define and develop \dgc in a modular fashion.
First, we define ``Generic Multi-instance Time-Lock Puzzle'' (\gm) which supports the multi-puzzle setting with different time intervals and efficiently verifying the correctness of solutions.
We then present an instantiation of \gm: the ``Generic Chained Time-Lock Puzzle'' (\gc) protocol.
We enhance \gm with client and server delegation to define \dgm.
Finally, we introduce \dgc which realizes \gm.
Our protocols use a standard time-lock puzzle scheme in a black-box manner.

We implemented and performed an in-depth evaluation of the protocols when generating and solving up to 10,000 puzzles. The implementation is open-source~\cite{repo}. 
To the best of our knowledge, it is the first time parties' overhead for a large number of puzzles is studied.
We found the computational overhead of \dgc to be negligible per puzzle and the cost of the smart contract is as low as 0.2 cents per puzzle. Applications of our (E)D-TLP include enhancing the scalability of (i) VDFs, by securely delegating VDFs computations to reduce the computational load on resource-constrained parties, and (ii) proof of storage, by offloading computationally intensive tasks—such as solving sequential challenges—to resource-capable helpers, as detailed in Section \ref{sec:applications}. 


\section{Preliminaries}\label{sec:preliminaries}
\subsection{Notations and Assumptions}\label{subsec:notations-and-assumptions}
We define function $\parse(\omega, y)\rightarrow (a, b)$, with inputs $\omega$ and $y$ of at least $\omega$ bits, that parses $y$ into $a$ of length $|y|-\omega$ and $b$ of length $\omega$.
To ensure generality in the definition of our verification algorithms, we adopt notations from zero-knowledge proof systems~\cite{BlumSMP91,FeigeLS90}.
Let $R$ be an efficient binary relation which consists of pairs of the form $(stm, wit)$, where $stm$ is a statement and $wit$ is a witness.
Let $\mathcal{L}$ be the language (in $\mathcal{NP}$) associated with $R$, i.e., $\mathcal{L}=\{stm|\ \exists wit \text{ s.t. }$ $ R(stm, wit)=1 \}$.
A (zero-knowledge) proof for $\mathcal{L}$ allows a prover to convince a verifier that $stm\in \mathcal{L}$ for a common input $stm$ (without revealing $wit$).

In the formal definitions, we use $\Pr\left[\begin{array}{c}\mathsf{Exp}\\ \hline\mathsf{Cond}\\\end{array}\right]$, where $\mathsf{Exp}$ is an experiment that involves an adversary $\adv$, and $\mathsf{Cond}$ is the set of winning conditions for $\adv$.

We denote network delay by $\Upsilon$.
As shown in~\cite{GarayKL15,BadertscherMTZ17}, after an honestly generated block is sent to the blockchain network, there is a maximum time period after which it will be observed by honest parties in the blockchain's unaltered part (a.k.a.\ state), ensuring a high probability of consistency.
For further discussion on the network delay, please refer to Appendix~\ref{sec::network-delay}.
We assume that parties use a secure channel when they interact with each other off-chain.

\subsection{Symmetric-key Encryption Scheme}\label{subsec:symmetric-key-encryption-scheme}

A symmetric-key encryption scheme consists of three algorithms:
\begin{enumerate*}[label=(\arabic*)]
\item $\skegen(1^\lambda)\rightarrow k$ is a probabilistic algorithm that outputs a symmetric key $k$.
\item $\skeenc(k,m)\rightarrow c$ takes as input $k$ and a message $m$ in some message space and outputs a ciphertext $c$.
\item $\skedec(k,c)\rightarrow m$ takes as input $k$ and a ciphertext $c$ and outputs a message $m$.
\end{enumerate*}
We require that the scheme satisfies \emph{indistinguishability under chosen-plaintext attacks}~(IND-CPA).
We refer readers to Appendix~\ref{sec:SKE} for more details.

\subsection{Time-Lock Puzzle}\label{subsec:time-lock-encryption}
We restate the formal definition of a TLP~\cite{Rivest:1996:TPT:888615}.
Our focus will be on the RSA-based \tlp, due to its simplicity and foundational role to the majority of TLPs.

\begin{definition}\label{Def::Time-lock-Puzzle}
A \tlp is a scheme between a client \cl and a server \se, consisting of the following algorithms.
\begin{itemize}[leftmargin=.37cm]
\item[$\bullet$]\brokenuline{$\setup{\tlp}(1^\lambda,\Delta, S)\rightarrow (\pk,\sk)$}.
    A probabilistic algorithm run by \cl.
    It takes as input a security parameter, $1^\lambda$, time parameter $\Delta$ that specifies how long a message must remain hidden in seconds, and time parameter $S$, the maximum number of squaring a server (with the highest level of computation resources) can perform per second.
    It outputs public and private keys $(\pk,\sk)$.
\item[$\bullet$]\brokenuline{$\gen{\tlp}(s, \pk, \sk)\rightarrow {p}$}.
    A probabilistic algorithm run by \cl.
    It takes as input solution $s$ and $(\pk,\sk)$ and outputs puzzle $p$.
\item[$\bullet$]\brokenuline{$\solve{\tlp}(\pk, p)\rightarrow s$}.
    A deterministic algorithm run by \se.
    It takes as input $\pk$ and $p$ and outputs solution $s$.
\end{itemize}

\noindent \tlp meets completeness and efficiency properties:
\begin{itemize}
\item[$*$] \textit{Completeness}.
    For any honest \cl and \se, $\solve{\tlp}(\pk$,\\ $ \gen{\tlp}(s, \pk, \sk))=s$.
\item[$*$] \textit{Efficiency}.
    The run-time of $\solve{\tlp}(\pk, p)$ is upper-bounded by a polynomial $\poly(\Delta,\lambda)$.
\end{itemize}
\end{definition}

The security of a TLP requires that the puzzle's solution stay confidential from all adversaries running in parallel within the time period, $\Delta$.
This is formally stated in Definition~\ref{def::TLP-sec}.

\begin{definition}\label{def::TLP-sec} A TLP is secure if for all $\lambda$ and $\Delta$, all probabilistic polynomial time (PPT) adversaries $\adv\coloneqq(\adv_1,\adv_2)$ where $\adv_1$ runs in total time $O(\textit{poly}(\Delta,\lambda))$ and $\adv_2$ runs in time $\delta(\Delta)<\Delta$ using at most polynomial number  $\xi(\Delta)$ of parallel processors, there is a negligible function $\mu(\cdot)$, such that:

\[ \Pr\left[%
\begin{array}{l}%
\pk, \sk \leftarrow \setup{\tlp}(1^\lambda,\Delta, S)\\
s_0,s_1,\text{state} \leftarrow \adv_1(1^\lambda,\pk, \Delta)\\
\rb\pick \bin\\
 p\leftarrow\gen{\tlp}(s_\rb, \pk, \sk)\\
\hline
 \rb \leftarrow \adv_2(\pk, p,\text{state}) \\

\end{array}
\right]\leq \frac{1}{2}+\negl\]
\end{definition}

By definition, TLPs are sequential functions.
Their construction requires that a sequential function, such as modular squaring, is invoked iteratively a fixed number of times.
The sequential function and iterated sequential functions notions, in the presence of an adversary possessing a polynomial number of processors, are formally defined in~\cite{BonehBBF18}.
We restate the definitions in Appendix~\ref{sec::sequential-squaring}.

Although not required by the definition, a desirable property of a TLP is that generating a puzzle much faster than solving it.
\citeauthor{Rivest:1996:TPT:888615}~\cite{Rivest:1996:TPT:888615} use a trapdoor RSA construction to meet the definition and achieve this property.
The security of the RSA-based TLP relies on the hardness of the factoring problem, the security of the symmetric key encryption, and the sequential squaring assumption.
This construction is detailed in Appendix~\ref{sec:RSA-based-TLP}.
We restate its formal definition below and refer readers to~\cite{Abadi-C-TLP} for the proof.

\begin{theorem}\label{theorem::R-LTP-Sec}
Let $N$ be a strong RSA modulus and $\Delta$ be the period for which the solution stays private.
If the sequential squaring holds, factoring $N$ is a hard problem and the symmetric-key encryption is semantically secure, then the RSA-based TLP scheme is a secure TLP\@.
\end{theorem}
\subsection{Multi-instance Time Lock Puzzle}\label{subsec:multi-instance-time-lock-puzzle}
\citeauthor{Abadi-C-TLP}~\cite{Abadi-C-TLP} proposed the Chained Time-Lock Puzzle~(C-TLP) to address the setting in which a server receives multiple puzzle instances separated by a fixed interval.
It encodes information required to solve each puzzle in the preceding puzzle's solution.
Only one puzzle in the chain needs to be solved at any one time, which, for $z$ puzzles, reduces computations by a factor of $\frac{z+1}{2}$. 
The scheme also introduces hash-based proofs of the correct solutions.
To date, it remains the most efficient time-lock puzzle for the multi-puzzle setting but it works only for fixed intervals.

\subsection{Commitment Scheme}\label{subsec:commit}

A commitment scheme involves two parties: \emph{sender} and \emph{receiver} and two phases: \emph{commit} and \emph{open}.
In the commit phase, the sender commits to a message $x$ as $\Com(x,r)=\Com_x$ using secret value $r\pick \bin^\lambda$.
The commitment $\Com_x$ is sent to the receiver.
In the open phase, the sender sends the opening $(x,r)$ to the receiver who verifies its correctness: $\Ver(\Com_x,(x,r))\qequal1$ and accepts if the output is $1$.
A commitment scheme must satisfy two properties:
\begin{enumerate*}[label=(\arabic*)]
    \item \textit{hiding}: it is infeasible for an adversary (the receiver) to learn any information about the committed message $x$ until the commitment $\Com_x$ is opened, and
    \item \textit{binding}: it is infeasible for an adversary (the sender) to open a commitment to different values than those used in the commit phase, meaning it is infeasible to find $(x', r')$ with $x\neq x'$ \textit{s.t.} $\Ver(\Com_x, (x,r))=\Ver(\Com_x, (x',r'))=1$, .
\end{enumerate*}
Efficient non-interactive commitment schemes exist in both the standard model, e.g., Pedersen scheme~\cite{Pedersen91}, and the random oracle model using the well-known hash-based scheme such that committing is: $\hash(x||r)=\Com_x$ and $\Ver(\Com_x, {x})$ requires checking: $\hash(x||r)\qequal\Com_x$, where $\hash:\bin^*\rightarrow \bin^\lambda$ is a collision-resistant hash function; i.e., the probability to find $x$ and $x'$ such that $\hash(x)=\hash(x')$ is negligible in the security parameter, $\lambda$.

\subsection{Smart Contract}\label{subsec:SC}
A smart contract is a computer program that encodes the terms and conditions of an agreement between parties and often contains a set of variables and functions.
Smart contract code is stored on a blockchain and is executed by the miners who maintain the blockchain.
When (a function of) a smart contract is triggered by an external party and the associated fee (termed \textit{gas}) is paid, every miner executes the smart contract.
The correctness of the program execution is guaranteed by the security of the underlying blockchain.

In this work, we use a standard smart contract to retain a pre-defined deposit amount, check if delegated work was completed correctly and on time, and distribute the deposit among parties accordingly.
Its role can also be played by any semi-honest server (including a semi-honest ``Trusted Execution Environment'' \cite{PassST17}), even one with minimal computational resources.

\section{Generic Multi-instance TLP}\label{sec::Generic-Multi-instance-Time-lock-Puzzle}

Although the Chained TLP~(\ctlp)~\cite{Abadi-C-TLP} can efficiently deal with multiple puzzles, it cannot handle time intervals of different sizes.
We address this limitation by formally defining the Generic Multi-instance Time-Lock Puzzle (\gm) and its concrete instantiation Generic Chained Time-Lock Puzzle (\gc).

\subsection{Definition and Security Properties}\label{Section::Multi-instance-Time-lock Puzzle-Definition}

\begin{definition}[Generic Multi-instance TLP]
A \gm is a scheme between a client \cl, a server \se, and a public verifier \ve, that consists of algorithms:

\begin{itemize}[leftmargin=.45cm]
\item[$\bullet$]\brokenuline{$\setup{\gmp}(1^\lambda, \intervalvec, S, z)\rightarrow (\pk,\sk)$}.
A probabilistic algorithm run by \cl.
It takes as input security parameter $1^\lambda$, a vector of time intervals $ \intervalvec=[\bar{\Delta}_1$, $\ldots$, $\bar{\Delta}_z]$ (such that the $j$-th solution is not found before $\sum_{i=1}^j \bar{\Delta}_i$), $S$, and the total number of puzzles $z$.
It outputs a set of public \pk and private \sk parameters.
    
\item[$\bullet$]\brokenuline{$\gen{\gmp}( \messvec, \pk, \sk)\rightarrow \hat{p}$}.
A probabilistic algorithm run by \cl.
It takes as input a message vector $\messvec=[m_1$, $\ldots$, $m_z]$ and $(\pk,\sk)$ and outputs $\hat{p}\coloneqq(\puzzvec$, $\commvec)$, where \puzzvec is a puzzle vector and \commvec is a public statement vector w.r.t.\ language \lang and \messvec.
Each $j$-th element in vectors \puzzvec and \commvec corresponds to a solution $s_j$ which itself contains $m_j$ (and possibly witness parameters).
\cl publishes \commvec and sends \puzzvec to $\se$.
     
\item[$\bullet$]\brokenuline{$\solve{\gmp}(\pk, \puzzvec)\rightarrow (\solsvec, \messvec)$}.
A deterministic algorithm run by \se.
It outputs solution vector \solsvec and plaintext message vector \messvec.
    
\item[$\bullet$]\brokenuline{$\prove{\gmp}(\pk, s_j)\rightarrow \pi_j$}.
A deterministic algorithm run by \se.
It takes as input \pk and a solution: $s_j\in \solsvec$ and outputs a proof $ \pi_j$ (asserting $m_j\in \lang$).
\se sends $m_j\in s_j$ and $\pi_j$ to \ve.

\item[$\bullet$]\brokenuline{$\verify{\gmp}(\pk, j, m_j, {\pi}_j, g_j)\rightarrow \ver_j\in \bin$}.
A deterministic algorithm run by \ve.
It takes as input \pk, $j$, $m_j$,  $ \pi_j$, and $g_j\in \commvec$ and outputs $1$ if it accepts the proof or $0$ otherwise.
\end{itemize}

\noindent \gm satisfies completeness and efficiency:
\begin{itemize}
\item[$*$] \textit{Completeness}.
For any honest prover and verifier, $\forall j \in [z]$, it holds that:
\begin{itemize}
\item[$\bullet$] $\solve{\gmp}(\pk,[p_1$, $\ldots$, $p_j]) \rightarrow ([s_1$, $\ldots$, $s_j]$, $[m_1$, $\ldots$, $m_j])$.
\item[$\bullet$] $\verify{\gmp}(\pk$, $j$, $m_j$, $\pi_j$, $g_j)\rightarrow 1$.
\end{itemize}

\item[$*$] \textit{Efficiency}.
The run-time of algorithm $\solve{\gmp}(\pk$, $[p_1$, $\ldots$, $p_j])=([s_1$, $\ldots$, $s_j]$, $[m_1$, $\ldots$, $m_j])$ is bounded by a polynomial
$\poly(\sum_{i=1}^j \bar{\Delta}_i, \lambda)$, $\forall j\in [z]$.
\end{itemize}
\end{definition}
 
In the above, the prover is required to generate a witness/proof $\pi_j$ for the language $\lang\coloneqq\{stm\coloneqq(g_j$, $m_j)|\ R(stm$, $\pi_j)=1\}$.
The proposed definition is generalization of the definition provided in~\cite{Abadi-C-TLP}: 
\begin{enumerate*}[label=(\arabic*)]
    \item different size intervals between puzzles are allowed and
    \item a broad class of verification algorithms and schemes are supported, not just hash-based commitments. 
\end{enumerate*}
We allow \solve{\gmp} to output both the solution vector and the plaintext message vector, as the solution to a puzzle may be an encoded version of the plaintext. 
 
At a high level, \gm satisfies a solution's \emph{privacy} and \emph{validity}.
Solution's privacy requires the $j$-th solution to remain hidden from all adversaries that run in parallel in the period: $\sum_{i=1}^j \bar{\Delta}_i$.
The solution's validity states that it is infeasible for a PPT adversary to compute an invalid solution and pass the verification.
We formally define the two properties below.

\begin{definition}[Privacy]\label{Def::GM-TLP::Solution-Privacy}
A \gm is privacy-preserving if for all $\lambda$ and $ \intervalvec=[\bar{\Delta}_1,\ldots, \bar{\Delta}_z]$, any number of puzzle: $z\geq1$, any $j \in [z]$, any pair of randomized algorithms $\adv \coloneqq (\adv_1,\adv_2)$, where $\adv_1$ runs in time $O(\poly(\sum_{i=1}^j \bar{\Delta}_i,\lambda))$ and $\adv_2$ runs in time $\delta(\sum_{i=1}^j \bar{\Delta}_i)<\sum_{i=1}^j \bar{\Delta}_i$ using at most polynomial  $\xi(\max(\bar{\Delta}_1,\ldots, \bar{\Delta}_j))$ parallel processors, there exists a negligible function $\mu(\cdot)$, such that:

\[\Pr\left[ \begin{array}{l}
\pk, \sk \leftarrow\setup{\gmp}(1^\lambda, \intervalvec, S, z)\\
\messvec_0, \messvec_1, \text{state}\leftarrow \adv_1(1^\lambda,\pk,z)\\
\rb_j\pick \bin, \forall j \in [z]\\
\hat{p} \leftarrow \gen{\gmp}( \messvec_\rb, \pk, \sk)\\
\hline
b',k \leftarrow \adv_2(\pk, \hat{p},\text{state})\\
\text{s.t.}\ 
b' = \rb_k\\

\end{array} \right]\leq \frac{1}{2}+\negl\]

\noindent where $k \in [z]$, $\messvec_0= [m_{0,1}$,\ldots,$m_{0, z}]$, $\messvec_1= [m_{1, 1}$,\ldots,$m_{1, z}]$ and $ \messvec_\rb=[m_{\rb_1, 1}$, $\ldots$, $m_{\rb_z, z}]$.
\end{definition}

Definition~\ref{Def::GM-TLP::Solution-Privacy} ensures solutions appearing after the $j$-th solution remain hidden from the adversary with a high probability.
Moreover, similar to~\cite{Abadi-C-TLP,BonehBBF18,MalavoltaT19,garay2019}, it ensures privacy even if $\adv_1$ computes on the public parameters for a polynomial time, $\adv_2$ cannot find the $j$-th solution in time $\delta(\sum_{i=1}^j \bar\Delta_i)<\sum_{i=1}^j \bar\Delta_i$ using at most $\xi(\max(\bar\Delta_1,\ldots, \bar\Delta_j))$ parallel processors, with a probability significantly greater than $\frac{1}{2}$.
As shown in~\cite{BonehBBF18}, we can set $\delta(\bar\Delta)=(1-\epsilon)\cdot\bar\Delta$ for a small $\epsilon$, where $0<\epsilon<1$.

\begin{definition}[Solution-Validity]\label{Def::GM-TLP::Solution-Validity}
A \gm preserves solution validity, if for all $\lambda$ and $ \intervalvec=[\bar\Delta_1,\ldots, \bar\Delta_z]$, any number of puzzles: $z\geq1$, all PPT adversaries $\adv\coloneqq\left(\adv_1,\adv_2\right)$ that run in time $O(\poly(\sum_{i=1}^z \bar\Delta_i$, $\lambda))$ there is a negligible function $\mu(\cdot)$, such that:

\[ \Pr\left[
\begin{array}{l}
\pk, \sk\leftarrow\setup{\gmp}(1^\lambda, \intervalvec, S, z)\\
\messvec,\text{state} \leftarrow \adv_1(1^\lambda,\pk,  \intervalvec)\\
\hat{p} \leftarrow\gen{\gmp}({\messvec}, \pk, \sk)\\
\solsvec, \messvec \leftarrow\solve{\gmp}(\pk, \puzzvec)\\
\hline
j,  \pi', m_j \leftarrow \adv_2(\pk, \solsvec,  \hat{p},\text{state})\\
 
\text{s.t. } 
\verify{\gmp}(\pk, j, m_j, {\pi}', g_j)= 1\\
m_j\notin \lang \\

\end{array} 
\right]\leq \negl\]
where $ {\messvec}=[m_1,\ldots,m_z]$, and $g_j\in  \commvec\in  \hat{p}$.
\end{definition}

\begin{definition}[Security]\label{def::GC-TLP-security}
A \gm is secure if it satisfies solution-privacy and solution-validity, w.r.t.\ Definitions~\ref{Def::GM-TLP::Solution-Privacy} and~\ref{Def::GM-TLP::Solution-Validity}, respectively.
\end{definition}

\subsection{Strawman Solutions}\label{subsec::strawman}
Before presenting our solution, we sketch potential schemes that provide variable intervals $\bar\Delta_j\neq \bar\Delta_{j+1}$ (or approximate them) and highlight their limitations to motivate our protocol.

One could use \tlp as a black-box by symmetrically encrypting each message and including the key in the preceding puzzle to construct a chain. 
This would have varied intervals but would require a fresh RSA key for each message imposing a computational and communication costs for \cl.
Another approach would be to employ \ctlp~\cite{Abadi-C-TLP} as a black-box by generating dummy puzzles to precede every real puzzle in the chain and approximate variable duration.
Specifically, \cl identifies the largest value $\Delta'$ that exactly divides all time intervals $\bar{\Delta}_j$.
To generate each puzzle $p_i$ containing message $m_i$, \cl computes $c_i=\frac{\bar{\Delta}_i}{\Delta}$, then uses \ctlp to produce $c_i$ chained puzzles: $c_i-1$ dummy puzzles followed by a puzzle containing $m_i$.
Given the example in \autoref{sec:introduction}, with two messages that are to be kept hidden for 24 hours and 40 hours, respectively, this would require a fixed interval of 8 hours and would produce 5 puzzles, in order: two dummy puzzles, a puzzle containing the message to be kept hidden for 24h, another dummy puzzle, and a puzzle containing the message to be kept hidden for 40h.
Although only one RSA modulus is needed, \cl has to perform additional modular exponentiation and transmit dummy puzzles to \se.
These costs are linear with the total number of dummy puzzles, which can explode when $\Delta'$ has to be very small to accurately capture the distinct intervals.
For example, combining a 99 hour puzzle and a 100 hour puzzle would require 100 total puzzles, a 50x increase in communication costs.

The protocol described in the next section offers a high level of flexibility (w.r.t.\ the size of intervals) without suffering from the aforementioned downsides.

\subsection{Generic Chained Time-lock Puzzle (\gc)}\label{subsec:gc-tlp-protocol}

We present \gc that realizes \gm.
At a high level, \gc lets a client \cl encode a vector of messages $\messvec=[m_1,\ldots,$ $m_z]$ so a server \se learns each message $m_j$ at time $t_j\in \timevec, \forall j \in [z]$.
In this setting, \cl is available only at an earlier time $t_0$, where $t_0< t_1$.

In \gc, we use the chaining technique, similar to \ctlp, and define the interval between puzzles as $\bar\Delta_j=t_j-t_{j-1}$, $ \intervalvec=[\bar\Delta_1,\ldots, \bar\Delta_z]$.
In the setup, for each $\bar\Delta_j$, \cl generates a set of parameters, including a time-dependent parameter $\ex_j$, which allows the message to be encoded with its corresponding time interval.
In the puzzle generation phase, \cl uses the above parameters to encrypt each message $m_j$.
The information required for decrypting $m_j$ is contained in the ciphertext of message $m_{j-1}$.
\se must solve the puzzles sequentially in the specified order, starting with extracting message $m_1$, which uniquely has the values required to solve it publicly available.
Below, we present \gc in detail.

\begin{enumerate}[leftmargin=.5cm]

\item \brokenuline{$\setup{\gmp}(1^\lambda, \intervalvec, S, z)\rightarrow (\pk,\sk)$}.
\label{GC-TLP::step::setup}
Involves \cl.  
    \begin{enumerate}[leftmargin=5mm]
    \item compute $N=q_1\cdot q_2$, where $q_i$ are large randomly chosen prime number.
    Next, compute Euler's totient function of $N$, as: $\phi(N)=(q_1-1)\cdot(q_2-1)$.

    \item $\forall j \in [z]$, set $T_j=S\cdot\bar\Delta_j$ as the number of squarings needed to decrypt an encrypted message $m_j$ after message $m_{j-1}$ is revealed, where $S$ is the maximum number of squarings modulo $N$ per second the strongest server can perform, and set $\timeexpvec=[T_1$, $\ldots$, $T_z]$.

    \item\label{GC-TLP::compute-a-values} compute values  $\ex_j=2^{T_j}\bmod \phi(N), \forall j \in [z]$, which yields vector $\expvec=[\ex_1,\ldots, \ex_z]$.

    \item pick fixed size random generators: $r_j\pick\mathbb{Z}^*_N,\forall j\in [z+1]$ with $|r_j|=\omega_1$, and set $\randvec=[r_2,\ldots,r_{z+1}]$.
    Generate $z$ symmetric-key encryption keys: $\keyvec=[k_1,\ldots,k_z]$ and pick $z$ fixed size sufficiently large random values: $\witnvec=[d_1,\ldots,d_z]$, with $|d_j|=\omega_2$.

    \item set public key as $\pk\coloneqq(\textit{aux}$, $N$, \timeexpvec, $r_1$, $\omega_1$, $\omega_2)$ and secret key as $\sk \coloneqq (q_1$, $q_2$, \expvec, \keyvec, \randvec, $\witnvec)$, where $\textit{aux}$ contains a hash function's description and the size of the random values.
    Output \pk and \sk.
    \end{enumerate}

\item \brokenuline{$\gen{\gmp}(\messvec, \pk, \sk)\rightarrow \hat{p}$}.
\label{GC-TLP::step::Generate-Puzzle}
Involves \cl. 
Encrypt the messages $\forall j \in [z]$:
    \begin{enumerate}[leftmargin=5mm]
    \item set $\pk_j \coloneqq (N,T_j, r_j)$ and $\sk_j\coloneqq(q_1, q_2, \ex_j, k_j)$.
    Recall $r_1$ is in \pk; for $j > 1$, $r_j$ is in \randvec.
    \item generate a puzzle by calling the puzzle generator algorithm (in the RSA-based puzzle scheme):\\ 
    $\gen{\tlp}(m_j||d_j||r_{j + 1},\pk_j, \sk_j) \rightarrow p_j$
    \item commit to each message: $\hash(m_j||d_j)=g_j$ and output $g_j$.
    \item output ${p}_j$ as puzzle. 
    \end{enumerate}
This phase yields vectors of puzzles: $\puzzvec=[ {p}_1,\ldots,{p}_z]$ and commitments: $\commvec=[g_1,\ldots,g_z]$.
Let $\hat p\coloneqq(\puzzvec, \commvec)$.
All public parameters and puzzles are given to the server at time $t_0 < t_1$, where $\bar\Delta_1=t_1-t_0$, and $\commvec$ is sent to the public.
This algorithm can be easily parallelized over the values of $j$.

\item\brokenuline{$\solve{\gmp}(\pk, \puzzvec)\rightarrow (\solsvec, \messvec)$}.
\label{GC-TLP::step::Solve-Puzzle}
Involves \se. 
Decrypt the messages in order from $j=1,\forall j\in [z]$:
    \begin{enumerate}[leftmargin=5mm]
    \item if $j=1$, $r_1\in \pk$; otherwise, $r_j$ is obtained by solving and parsing the previous puzzle.
    
    \item set $\pk_j\coloneqq(N,T_j,r_j)$.
    
    \item call the puzzle solving algorithm in the \tlp scheme: \\
    $\solve{\tlp}(\pk_j, p_j)\rightarrow x_j$, where $p_j \in \puzzvec$.
    
    \item parse $x_j=m_j||d_j||r_{j+1}$ as:
        \begin{enumerate}
        \item $\parse(\omega_1, m_j||d_j||r_{j+1})\rightarrow (m_j||d_j, r_{j+1})$.
        \item $\parse(\omega_2, m_j||d_j)\rightarrow (m_j, d_j)$.
        \item output $s_j =  m_j||d_j$ and $m_j$.
        \end{enumerate}
    \end{enumerate}

    By the end of this phase, vectors of solutions $\solsvec=[s_1,\ldots,s_z]$ and plaintext messages $\messvec=[m_1,\ldots,m_z]$ are generated. 

\item \brokenuline{$\prove{\gmp}(\pk,s_j)\rightarrow \pi_j$}.
\label{GC-TLP::step::prove-}
Involves \se.
    \begin{enumerate}[leftmargin=5mm]
    \item parse $s_j$ into $(m_j,d_j)$.
    \item send $m_j$ and $\pi_j=d_j$ to the verifier.
    \end{enumerate}

\item \brokenuline{$\verify{\gmp}(\pk, j, m_j, \pi_j, g_j)\rightarrow a_j \in \bin$}.
\label{GC-TLP::step::verify-}
Involves the verifier, e.g., the public or \cl.
    \begin{enumerate}[leftmargin=5mm]
    \item verify the commitment: $\hash(m_j, \pi_j)\qequal g_j$.
    \item if verification succeeds, accept the solution and output $1$; otherwise, output $0$.
    \end{enumerate}

\end{enumerate}

Note that \cl can extend a puzzle chain later, only retaining \pk and \sk, without having to be online in the intervening time.
A server must fully solve the existing chain before starting on new puzzles.
Appendix~\ref{sec::extending-a-puzzle-chain} specifies the steps required.

\begin{theorem}\label{theorem::GC-TLP-sec}
\gc is a secure multi-instance time-lock puzzle, w.r.t.\ Definition~\ref{def::GC-TLP-security}.
\end{theorem}

\subsection{Security Analysis of \gc}\label{sec::GC-Proof}

In this section, we prove the security of \gc, i.e.,  Theorem~\ref{theorem::GC-TLP-sec}.

\begin{proof} 
The proof has similarities with that of \ctlp~\cite{Abadi-C-TLP}.
Nevertheless, it has significant differences.
First, we prove a solver cannot find the parameters needed to solve $j$-th puzzle, without solving the previous puzzle, $(j-1)\text{\small{-th}}$.

\begin{lemma}\label{lemma::Next-Generator-Privacy}
Let $N$ be a large RSA modulus, $k$ be a random key for symmetric-key encryption, and $\lambda=\log_2(N)=\log_2(k)$ be the security parameter.
In \gc, given puzzle vector \puzzvec and public key \pk, an adversary $\adv\coloneqq(\adv_1,\adv_2)$, defined in Section~\ref{Section::Multi-instance-Time-lock Puzzle-Definition}, cannot find the next group generator $r_j$, where $r_j \pick \mathbb{Z}^*_N$ and $j\geq1$, significantly in time smaller than $T_j=\delta(\sum_{i=1}^{j-1} \bar\Delta_i)$, except with a negligible probability in the security parameter, \negl.
\end{lemma}

\begin{proof}
 The next generator, $r_j$, is picked uniformly at random from $\mathbb{Z}^{*}_{N}$ and is encrypted along with the $(j-1)\text{\small{-th}}$ puzzle solution, $s_{j-1}$.
 For the adversary to find $r_j$ without performing enough squaring, it has to either break the security of the symmetric-key scheme, decrypt the related ciphertext $s_{i-1}$ and extract the random value from it, or guess $r_j$ correctly.
 In both cases, the adversary's probability of success is negligible \negl.
\end{proof}
  
Next, we prove \gc preserves a solution's privacy with regard to Definition~\ref{Def::GM-TLP::Solution-Privacy}.

\begin{theorem}\label{Solution-Privacy}
 Let $\intervalvec=[\bar\Delta_1,\ldots, \bar\Delta_z]$ be a vector of time parameters and $N$ be a strong RSA modulus.
 If the sequential squaring assumption holds, factoring $N$ is a hard problem, $\hash(\cdot)$ is a random oracle and the symmetric-key encryption is  IND-CPA secure, then \gc (which encodes $z$ solutions) is a privacy-preserving \gm according to Definition~\ref{Def::GM-TLP::Solution-Privacy}.
\end{theorem}
\begin{proof}
For adversary $\adv\coloneqq(\adv_1,\adv_2)$, where $\adv_1$ runs in total time $O(\poly(\sum_{i=1}^z \bar\Delta_i$, $\lambda))$, $\adv_2$ runs in time $\delta(\sum_{i=1}^j \bar\Delta_{i})<\sum_{i=1}^j \bar\Delta_{i}$ using at most $\xi(\max(\bar\Delta_1,\ldots, \bar\Delta_j))$ parallel processors, and $j\in [z]$, we have two cases.
In case $z=1$: to find $s_1$ earlier than $\delta(\bar\Delta_1)$, it has to break the original TLP scheme~\cite{Rivest:1996:TPT:888615}, as the two schemes are identical, yet TLP is known to be secure from Theorem~\ref{theorem::R-LTP-Sec}.
In case $z>1$: to find $s_j$ earlier than $T_j=\delta(\sum_{ i=1}^{j} \bar\Delta_i)$, it has to either find one of the prior solutions earlier than its predefined time, which would require it to break the TLP scheme again, or to find the related generator, $r_j$, earlier than it is supposed to yet its success probability is negligible due to Lemma~\ref{lemma::Next-Generator-Privacy}.
 
The adversary may want to find partial information of the commitment, $g_j$, pre-image (which contains the solution) before solving the puzzle.
But, this is infeasible for a PPT adversary, given that $\hash(\cdot)$ is a random oracle.
We conclude that \gc is a {privacy-preserving} generic multi-instance time-lock puzzle scheme.
\end{proof}

Next, we present the theorem and proof for \gc solution's validity, w.r.t.\ Definition~\ref{Def::GM-TLP::Solution-Validity}.

\begin{theorem} \label{GC-TLP::Solution-Validity}
 Let $\hash(\cdot)$ be a hash function modeled as a random oracle.
 Then, \gc preserves a solution's validity, according to Definition~\ref{Def::GM-TLP::Solution-Validity}.
\end{theorem}
\begin{proof}
 The proof is reduced to the security of the hash-based commitment.
 Given the commitment $g_j=\hash(m_j$, $d_j)$ and an opening $\pi\coloneqq(m_j$, $d_j)$, for an adversary to break solution validity, it must find $(m'_j,d'_j)$, such that $\hash(m'_j,d'_j)=g_j$, where $m_j\neq m'_j$, i.e., a collision.
 However, this is infeasible for a PPT adversary, as $\hash(\cdot)$ is collision-resistant in the random oracle model.
\end{proof}
\renewcommand{\qedsymbol}{$\blacksquare$}
As indicated in the proofs of Theorems~\ref{Solution-Privacy} and~\ref{GC-TLP::Solution-Validity}, \gc preserves the privacy and validity of a solution, respectively.
Hence, per Definition~\ref{def::GC-TLP-security}, \gc is a \emph{secure} generic multi-instance time-lock puzzle scheme.
\end{proof}


\section{Delegated Time-lock Puzzle}\label{sec:Delegated-Time-lock-Puzzle}

We introduce the notion of Delegated TLP (\dgm), which is an enhancement of \gm.
We then proceed to present Efficient Delegated Time-Lock Puzzle (\dgc), a protocol embodying the properties of \dgm.
The key features that \dgm offers include:

\begin{enumerate}[leftmargin=.48cm]

\item \emph{Client-side delegation}, that allows \cl to delegate setup and puzzle generation phases to a third-party helper \tpc (modeled as a semi-honest adversary) while preserving the privacy of its plaintext messages.
\item \emph{Server-side delegation}, which enables \se to offload puzzle solving to another third-party helper \tps(modeled as a rational adversary with a utility function equal to the sum of payments made to the helper minus the helper's computational costs) while ensuring the privacy and integrity of the plaintext solutions.
\item \emph{Accurate-time solution delivery}, that guarantees each solution is delivered before a pre-defined time.
\end{enumerate}
To capture the latter property, we parameterize \dgm with an \textit{upper bound} and a \textit{fair payment algorithm}.

An upper bound explicitly specifies the time by which a helper must find a solution.
This parameter is not explicit in existing TLPs: the only time parameter is the lower bound.
To establish the upper bound, we introduce the ``Customized Extra Delay Generating'' function, which produces a time parameter, denoted as $\Psi$.
This specifies the duration after $\Delta$ when the helper is expected to discover the solution to the puzzle.

\begin{definition}[Customized Extra Delay Generating Function] \label{def::Customized-extra-delay-predicate}
Let $ToC$ be the type of computational step a scheme relies on, $S$ be the maximum number of steps a server with the best computational resources can compute per second, $\Delta$ be the interval in seconds a message must remain private, and $aux_{ID}$ be auxiliary data about a specific server identified with $ID$, e.g., its computational resources.
The Customized Extra Delay Generating function is defined as:
\[\ced(ToC, S, \Delta, aux_{ID})\rightarrow \Psi_{ID}\]
It returns $\Psi_{ID}$: the extra time in seconds solver $ID$ requires to solve the puzzle in addition to $\Delta$.

\end{definition}

A fair payment algorithm compensates a party only when it submits a valid solution before a pre-defined deadline.
Implementing fair payments is crucial to guarantee the timely delivery of a solution by a third party helper \tps in the context of delegated TLP.
Note that simply relying on an upper bound would not guarantee that the solution will be delivered by a rational adversary that has already identified the solution.
This algorithm will be incorporated into a smart contract.

\subsection{Definition and Security Properties}\label{sec::DGM-TLP-def}

We initially present the syntax of \dgm in Definition~\ref{def::dgm-syntax} and then present its security properties, in Definitions~\ref{Def::DGM-TLP::Solution-Privacy},~\ref{Def::DGM-TLP::Solution-Validity}, and~\ref{def::dgm-tlp-fair-payment}.

\begin{definition}[Delegated Time-Lock Puzzle]\label{def::dgm-syntax}
A \dgm is a scheme between a client \cl, a server \se, a pair of third-party helpers (\tpc, \tps), and a smart contract \scc that consists of algorithms:
\begin{itemize}[leftmargin=.43cm]
\item[$\bullet$]\brokenuline{$\cl.\setup{\mtdgm}(1^\lambda, \intervalvec) \rightarrow (\cpk,\csk)$}.
A probabilistic algorithm run by \cl. It generates a set of public \cpk and private \csk parameters, given $\intervalvec=[\bar{\Delta}_1$, $\ldots$, $\bar{\Delta}_z]$, where $\sum_{i=1}^j \bar{\Delta}_i$ is the period after which the $j$-th solution is discovered.
It appends \intervalvec to \cpk and outputs $(\cpk,\csk)$.
\cl sends \cpk and \csk to \se.

\item[$\bullet$] \brokenuline{$\cl.\delegate{\mtdgm}(\messvec, \cpk, \csk)\rightarrow (\messvec^*, t_0)$}.
A probabilistic algorithm run by \cl. It encodes each element of $ \messvec=[m_1$, $\ldots$, $m_z]$ to produce a vector of encoded messages $\messvec^*=[m^*_1$,$\ldots$,$m^*_z]$ and time point $t_0$ when puzzles are provided to \tps.
\cl publishes $t_0$ and sends $\messvec^*$ to \tpc.

\item[$\bullet$] \brokenuline{$\se.\delegate{\mtdgm}(\ced$, $ToC$, $S$, \intervalvec, $aux$, $\adr_\tps$, $t_0$, $\Upsilon$, $\coinsvec)$ $\rightarrow \adr_\scc$}.
A deterministic algorithm executed by \se.
It runs $\ced(ToC$, $S$, $ \bar\Delta_i$, $aux)\rightarrow \Psi_i$, which uses the auxiliary information about the helper $aux$ to determine the acceptable delay for every $\bar\Delta_i$ in \intervalvec.
It generates a smart contract \scc, deployed into the blockchain with address $\adr_\scc$.
\se deposits
$\coins=\sum_{i=1}^z\coins_i$ coins into \scc, where $\coins_i \in \coinsvec$ is the amount of coins paid to $\tps$ for solving the $i$-th puzzle.
It registers $t_0, \extravec=[\Psi_1$, $\ldots$, $\Psi_z]$, $\adr_\tps$, and $\Upsilon$ in \scc.
It returns $\adr_\scc$, the address of the deployed \scc.
\se publishes $\adr_\scc$.

\item[$\bullet$]\brokenuline{$\setup{\mtdgm}(1^\lambda, \intervalvec, S, z)\rightarrow (\pk,\sk)$}.
A probabilistic algorithm run by \tpc. It outputs a set of public $\pk$ and private $\sk$ parameters.
\tpc sends \pk to \tps.

\item[$\bullet$] \brokenuline{$\gen{\mtdgm}({\messvec^*}, \cpk, \pk, \sk, t_0)\rightarrow \hat p$}.
A probabilistic algorithm run by \tpc. It outputs $\hat{p}\coloneqq(\puzzvec, \commvec)$, where \puzzvec is a puzzle vector, \commvec is a public statement vector w.r.t.\ language \lang and $\messvec^*$.
Each $j$-th element in vectors \puzzvec and \commvec corresponds to a solution $s_j$ that itself consists of $m_j^*$ (and possibly witness parameters).
\tpc sends \commvec to $\scc$.
It also sends \puzzvec to \tps at time $t_0$.

\item[$\bullet$]\brokenuline{$\solve{\mtdgm}(\extravec, aux, \pk, \puzzvec, \adr_\scc, \coinsvec', \coinsvec)\rightarrow (\solsvec, q)$}.
A deterministic algorithm run by \tps, in which $\coinsvec'$ is the vector of payments \tps expects.
It checks \coinsvec and \extravec.
If it does not agree on these parameters, it sets $q = 0$ and outputs $(, q)$ and the rest of the algorithms will not be run.
Else, it sets $q = 1$, finds the solutions, and outputs a vector \solsvec of solutions and $q$.

\item[$\bullet$] \brokenuline{$\prove{\mtdgm}(\pk, s_j)\rightarrow \pi_j$}.
A deterministic algorithm run by \tps. It outputs a proof $\pi_j$ asserting $m_j^*\in \lang$.

\item[$\bullet$]\brokenuline{$\register{\mtdgm}(s_j, \pi_j, \adr_\scc)\rightarrow t_j$}.
A deterministic algorithm run by \tps. It registers $m_j^*$ and $\pi_j$ in \scc, where $m_j^*\in s_j$.
It receives registration time $t_j$ from \scc.
It outputs $t_j$.

\item[$\bullet$] \brokenuline{$\verify{\mtdgm}(\pk, j, m_j^*, \pi_j, g_j, \Psi_j, t_j, t_0, \intervalvec, \Upsilon)\rightarrow \ver_j$}.
A deterministic algorithm run by \scc. It checks whether (i) proof $\pi_j$ is valid and (ii) the $j$-th solution was delivered on time.
If both checks pass, it outputs $\ver_j=1$; otherwise, it outputs $\ver_j=0$.

\item[$\bullet$] \brokenuline{$\pay{\mtdgm}(\ver_j, \adr_\tps, \coinsvec, j)\rightarrow \paid_j$}.
A deterministic algorithm run by \scc. It pays out based on the result of verification $\ver_j$.
If $\ver_j=1$, it sends $\coins_j \in \coinsvec$ coins to an account with address $\adr_\tps$ and sets $\paid_j=1$; otherwise, it sets $\paid_j=0$.
It outputs $\paid_j$.

\item[$\bullet$] \brokenuline{$\retrieve{\mtdgm}(\pk, \csk, m_j^*)\rightarrow m_j$}.
A deterministic algorithm run by \se. It retrieves message $m_j$ from $m_j^*$ and outputs $m_j$.
\end{itemize}

\noindent\dgm satisfies completeness and efficiency properties.
\begin{itemize}[leftmargin=.43cm]
\item[*] \textit{Completeness}.
For honest \cl, \se, \scc, \tpc, and \tps, $\forall j\in [z]$, it holds that:

\begin{itemize}[leftmargin=0cm]
\item[$\bullet$] $\solve{\mtdgm}(\extravec, aux$, \pk, $[p_1$, $\ldots$, $p_j])$, $\adr_\scc$, $\coinsvec'$, $\coinsvec) \rightarrow([s_1$, $\ldots$, $s_j], 1)$.
\item[$\bullet$] $\verify{\mtdgm}(\pk, j, m_j^*, \pi_j, g_j, \Psi_j, t_j, t_0, \intervalvec, \Upsilon)\rightarrow \ver_j = 1$.
\item[$\bullet$] $\pay{\mtdgm}(ver_j, \adr_\tps, \coinsvec, j)\rightarrow \paid_j=1$.
\item[$\bullet$] $\retrieve{\mtdgm}(\pk, \csk, m_j^*)\rightarrow m_j$, where $m_j\in \messvec$.
\end{itemize}

\item[*]\textit{Efficiency}.
The run-time of $\solve{\mtdgm}(\extravec$, $aux$, $\pk$, $[p_1$, $\ldots$, $p_j]$, $\adr_\scc$, $\coinsvec'$, $\coinsvec) \rightarrow ([s_1$, $\ldots$, $s_j]$, $1)$ is bounded by
a polynomial $\poly(\sum_{i=1}^j \bar\Delta_i, \lambda)$.
\end{itemize}
\end{definition}

\dgm's definition, similar to \gm, supports a solution's privacy and validity with the addition that privacy requires plaintext messages $[m_1$, $\ldots$, $m_z]$ to remain hidden from \tpc and \tps.

In Case~\ref{def::solution-privacy-from-helper} in Definition~\ref{Def::DGM-TLP::Solution-Privacy}, we state that given the algorithms' transcripts, an adversary that picks a pair of plaintext messages cannot tell which message is used for the puzzle with a probability significantly greater than $\frac{1}{2}$. 
In Case~\ref{def::solution-privacy-from-solver} in Definition~\ref{Def::DGM-TLP::Solution-Privacy}, we formally state that the privacy of an encoded solution $m_j^*$ requires $j$-th encoded solution to remain hidden from all adversaries that run in parallel in period $\sum_{i=1}^j \bar{\Delta}_i$.

\begin{definition}[Privacy]\label{Def::DGM-TLP::Solution-Privacy}
A \dgm is privacy-preserving if for all $\lambda$ and $\intervalvec=[\bar\Delta_1, \ldots, $ $\bar\Delta_z]$, any number of puzzles: $z\geq1$, any $j \in [z]$ the following hold:

\begin{enumerate}[leftmargin=.5cm]

\item\label{def::solution-privacy-from-helper}
For any PPT adversary $\adv_1$, there exists a negligible function $\mu(\cdot)$, such that:
\[\Pr\left[\begin{array}{l}
\cpk, \csk \leftarrow\cl.\setup{\mtdgm}(1^\lambda, \intervalvec)\\
\messvec_0, \messvec_1,\text{state}\leftarrow\adv_1(1^\lambda, \cpk, z)\\
\rb_i \pick \bin,\forall i \in [j]\\
\messvec^*, t_0 \leftarrow\cl.\delegate{\mtdgm}(\messvec_\rb, \cpk, \csk)\\
\adr_\scc \leftarrow \se.\delegate{\mtdgm}(\ced, ToC, S,\intervalvec, \\\hphantom{\adr_\scc \leftarrow} aux, \adr_\tps, t_0, \Upsilon, \coinsvec)\\
\pk,\sk\leftarrow\tpc.\setup{\mtdgm}(1^\lambda, \intervalvec, S, z)\\
\hat{p}\leftarrow \gen{\mtdgm}(\messvec^*, \cpk, \pk, \sk, t_0)\\
\solsvec, q \leftarrow\solve{\mtdgm}(\extravec, aux, \pk, \puzzvec, \adr_\scc, \coinsvec',\\\hphantom{\solsvec, q \leftarrow} \coinsvec)\\
\hline
b',k\leftarrow\adv_1(\pk, \cpk, \messvec_0, \messvec_1, \messvec^*, \text{state},t_0, \ced,\\\hphantom{b',k}
 ToC, aux, \extravec, \intervalvec, \solsvec,\hat{p}, \coinsvec, \adr_\tps, \adr_\scc)\\

\text{s.t.}\
b' = \rb_k \\

\end{array}\right]\hspace{-.3mm}\leq\hspace{-1mm} \frac{1}{2}+\mu(\lambda)\hspace{-1mm}\]
\noindent where 
$\messvec_0 = [m_{0,1}$,\ldots,$m_{0,z}]$,
$\messvec_1 = [m_{1,1}$,\ldots,$m_{1,z}]$,
$\messvec_\rb = [m_{\rb_1,1}$, \ldots, $m_{\rb_z,z}]$.

\item\label{def::solution-privacy-from-solver}
For any two randomized algorithms $\adv \coloneqq (\adv_2,\adv_3)$,
where $\adv_2$ runs in time $O(\poly(\sum_{i=1}^j \bar{\Delta}_i,\lambda))$
and $\adv_3$ runs in time $\delta(\sum_{i=1}^j \bar{\Delta}_i)$ $<\sum_{i=1}^j \bar{\Delta}_i$ using at most $\xi(\max(\bar{\Delta}_1, \ldots, \bar{\Delta}_j))$ parallel processors, there is a negligible function $\mu(\cdot)$, such that:

\[\Pr\left[\begin{array}{l}

\cpk, \csk\leftarrow\cl.\setup{\mtdgm}(1^\lambda, \intervalvec)\\
\messvec_0^*, \messvec^*_1, \text{state}\leftarrow\adv_2(1^\lambda, \cpk, \csk, z)\\
\adr_\scc\leftarrow\se.\delegate{\mtdgm}(\ced, ToC,\\
\hphantom{\adr_\scc\leftarrow} S, \intervalvec, aux, \adr_\tps, t_0, \Upsilon, \coinsvec)\\
\pk, \sk\leftarrow\tpc.\setup{\mtdgm}(1^\lambda, \intervalvec, S, z)\\
\rb_i\pick \bin,\forall i \in [j]\\
\hat{p}\leftarrow\gen{\mtdgm}( \messvec^*_\rb, \cpk, \pk, \sk, t_0)\\

\midrule

b', k \leftarrow \adv_3(\pk, \cpk,\hat{p}, \text{state}, \extravec, \intervalvec, \adr_\tps,\\
\hphantom{b', k \leftarrow} \adr_\scc, aux, \ced, ToC, t_0)\\
\text{s.t.}\
b' = \rb_k \\
\end{array} \right]\hspace{-.3mm}\leq\hspace{-1mm} \frac{1}{2}+\negl\]
\end{enumerate}

\noindent where 
$\messvec^*_0= [m^*_{0,1}$,\ldots, $m^*_{0,j}]$,
$\messvec^*_0= [m^*_{1,1}$,\ldots, $m^*_{1,j}]$,
$\messvec^*_\rb=[m^*_{\rb_1,1}$, $\ldots$, $m^*_{\rb_j,j}]$.
\end{definition}

Intuitively, solution validity requires that a prover cannot persuade a verifier
\begin{enumerate*}[label=(\arabic*)]
    \item to accept a solution that is not equal to the encoded solution $m^*_j$ or
    \item to accept a proof that has been registered after the deadline,
\end{enumerate*}
except for a probability negligible in the security parameter.

\begin{definition}[Solution-Validity]\label{Def::DGM-TLP::Solution-Validity}
A \dgm preserves a solution validity, if for all $\lambda$ and $ \intervalvec=[\bar{\Delta}_1,\ldots, \bar{\Delta}_z]$, any number of puzzles: $z\geq1$, all PPT adversaries $\adv\coloneqq(\adv_1,\adv_2)$ that run in time $O(\poly(\sum_{i=1}^z \bar{\Delta}_i$, $\lambda))$ there is a negligible function $\mu(\cdot)$, such that:

\[ \Pr\left[
    \begin{array}{l}
\cpk, \csk\leftarrow\cl.\setup{\mtdgm}(1^\lambda, \intervalvec)\\
\messvec, \messvec^*, \text{state}\leftarrow\adv_1(1^\lambda,\pk, \intervalvec, z)\\
\adr_\scc \leftarrow \se.\delegate{\mtdgm}(\ced, ToC, S, \intervalvec, aux,\\
\hphantom{\adr_\scc \leftarrow}  \adr_\tps, t_0, \Upsilon, \coinsvec)\\
\pk,\sk\leftarrow\tpc.\setup{\mtdgm}(1^\lambda, \intervalvec, S, z)\\
\hat{p}\leftarrow\gen{\mtdgm}({\messvec^*}, \cpk, \pk, \sk, t_0)\\
\solsvec, q\leftarrow\solve{\mtdgm}(\extravec, aux, \pk, \puzzvec, \adr_\scc, \coinsvec', \coinsvec)\\
\midrule

j, \pi', m_j^*, t_j, t'_j\leftarrow\adv_2(\pk, \cpk, \solsvec, \hat{p},\text{state},csk, \extravec, \intervalvec, \\
\hphantom{ j, \pi', m_j^*,} aux, \coinsvec, \adr_\scc, \adr_\tps, \ced, ToC, t_0)\\
\text{s.t.}\
(con_1 \wedge con_2\wedge con_3)
\vee
(\neg con_1 \wedge con_4\wedge con_5)\\


\end{array}
   \right]\leq  \mu(\lambda)\]
where $\messvec=[m_1$, $\ldots$, $m_z]$, and $g_j\in \commvec\in \hat{p}$.
Each condition $con_i$ is defined as follows:
\begin{itemize}[leftmargin=3.2mm]
\item[$\bullet$] $con_1$: $j$-th solution delivered on time:
$t_j- t_0\leq\sum_{i=1}^j(\bar{\Delta}_i+\Psi_i+\Upsilon)$.
\item[$\bullet$] $con_2$: generates invalid proof $\pi'$: $m_j^*\notin \lang$.
\item [$\bullet$] $con_3$: the invalid proof $\pi'$ is accepted:\\
$\verify{\mtdgm}(\pk$, $j$, $m_j^*$, $\pi'$, $g_j$, $\Psi_j$, $t_j$, $t_0$, \intervalvec, $\Upsilon)=1$.
\item [$\bullet$]$con_4$: generates valid proof $\pi_j$:
$m_j^*\in \lang$.
\item [$\bullet$]$con_5$: a valid proof is accepted, despite late delivery:\\
$\verify{\mtdgm}(\pk$, $j$, $m_j^*$, $ \pi_j$, $g_j$, $\Psi_j$, $t'_j$, $t_0$, \intervalvec, $\Upsilon)=1$.
\end{itemize}
\end{definition}

Informally, fair payment states that the verifier pays the prover if and only if the verifier accepts the proof.

\begin{definition}[Fair Payment]\label{def::dgm-tlp-fair-payment} A
\dgm supports fair payment, if $\forall j \in [z]$ and any PPT adversary \adv, there is a negligible function $\mu(\cdot)$, such that:
\[\Pr\left[\begin{array}{l}
\cpk, \csk\leftarrow\cl.\setup{\mtdgm}(1^\lambda, \intervalvec)\\
\messvec,\text{state}\leftarrow\adv(1^\lambda, \cpk, z)\\
\messvec^*, t_0\leftarrow \delegate{\mtdgm}(\messvec, \cpk, \csk)\\
\adr_\scc\leftarrow\se.\delegate{\mtdgm}(\ced,ToC, S,\intervalvec, aux, \\
\hphantom{\adr_\scc\leftarrow} \adr_\tps, t_0, \Upsilon, \coinsvec)\\
\pk,\sk\leftarrow\tpc.\setup{\mtdgm}(1^\lambda, \intervalvec, S, z)\\
\hat{p}\leftarrow\gen{\mtdgm}( {\messvec^*}, \cpk, \pk, \sk, t_0)\\
s_j, t_j, \pi_j\leftarrow\adv(\cpk, \messvec, \text{state}, t_0, \ced, ToC, aux, \\
\hphantom{s_j, t_j, \pi_j} \extravec, \intervalvec,\coinsvec, \adr_\scc,  \adr_\tps, \pk, \puzzvec)\\
\ver_j\leftarrow\verify{\mtdgm}(\pk, j, m_j^*, \pi_j, g_j, \Psi_j, t_j, t_0, \intervalvec, \Upsilon)\\
\paid_j\leftarrow\pay{\mtdgm}(\ver_j, \adr_\tps, \coinsvec, j)\\

\hline
\ver_j\neq \paid_j\\

\end{array}\right]\leq \mu(\lambda)\]

\end{definition}

\begin{definition}[Security]\label{def::DGM-TLP-security} A \dgm is secure if it satisfies solution privacy, solution-validity, and fair payment, w.r.t.\ definitions~\ref{Def::DGM-TLP::Solution-Privacy},~\ref{Def::DGM-TLP::Solution-Validity}, and~\ref{def::dgm-tlp-fair-payment} respectively.
\end{definition}

\subsection{Efficient Delegated Time-locked Puzzle}\label{subsec:efficient-delegated-time-locked-puzzle}
At a high level, \dgc operates as follows.
During the setup, \cl encrypts all the plaintext solutions using symmetric-key encryption with key \csk, which is sent to \se.
In turn, \se picks a \tps and determines the extra time it needs to find each solution using $\ced(\cdot)$.
Next, \se constructs a smart contract \scc in which it specifies the expected delivery time for each solution.
It deploys \scc to the blockchain and deposits enough coins for $z$ valid solutions.

\cl sends the ciphertexts to \tpc who
\begin{enumerate*}[label=(\arabic*)]
    \item generates all required secret and public keys on \cl's behalf and
    \item builds puzzles on the ciphertexts (instead of the plaintext solutions in \gc).
\end{enumerate*}
It sends all puzzles to \tps who checks the deposit and the \scc parameters.

If \tps agrees to proceed, it solves each puzzle and generates a proof of the solutions' correctness.
It sends the solution and proof to \scc who checks if the solution-proof pair has been delivered on time and the solution is valid with the help of the proof.
If the two checks pass, \scc sends a portion of the deposit to \tps.

\begin{figure}[!b]
    \centering
    \includegraphics[width=\linewidth]{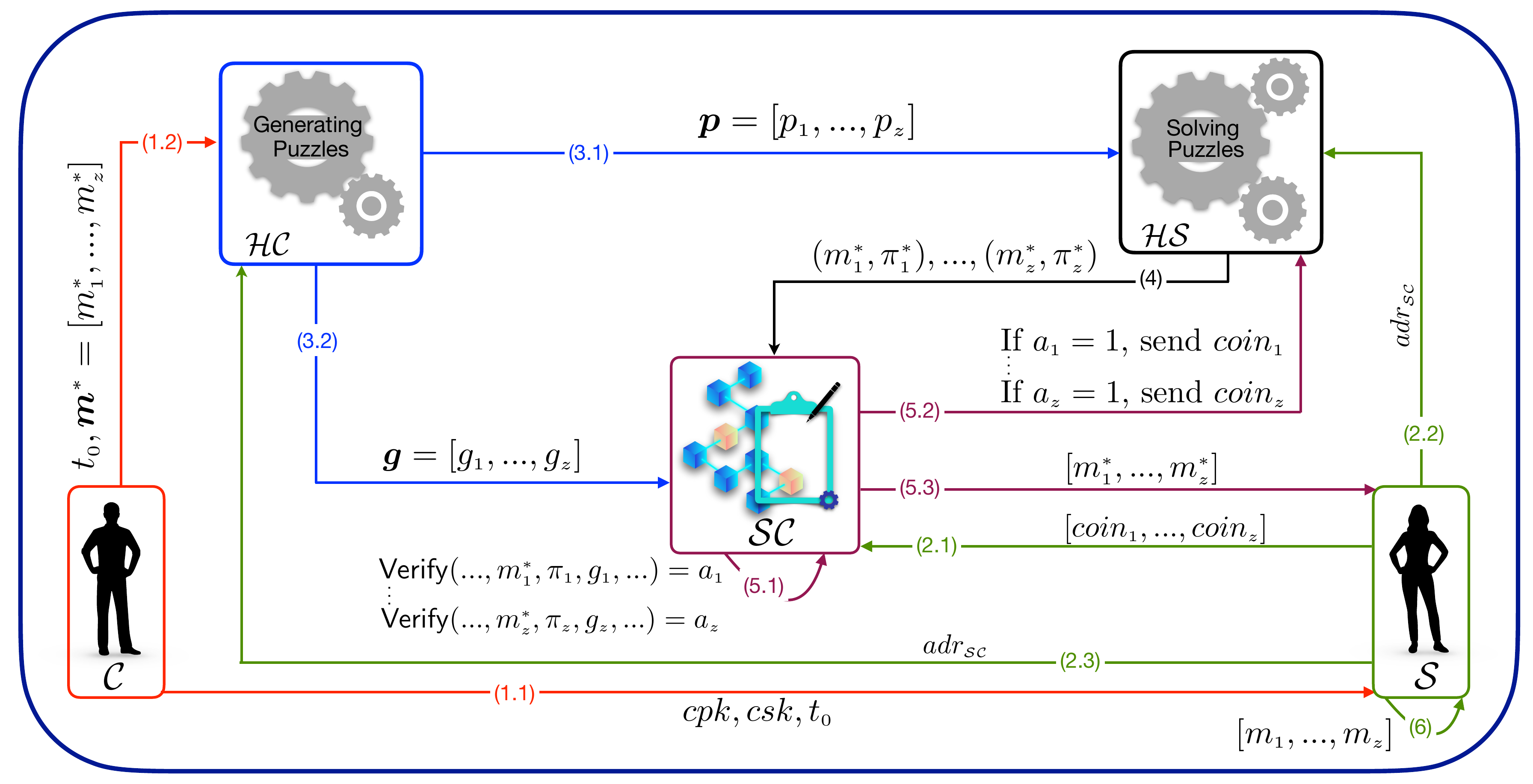}
    \caption{Outline of the interactions between parties in ED-TLP\xspace. \ensuremath{\mathcal{C}}\xspace is the client, \ensuremath{\mathcal{S}}\xspace is the server, \ensuremath{\mathcal{HC}} is \ensuremath{\mathcal{C}}'s helper, \ensuremath{\mathcal{HS}} is \ensuremath{\mathcal{S}}'s helper, and \ensuremath{\mathcal{SC}} is the smart contract.}\label{fig:ED-TLP-workflow}
    \Description{The diagram maps the interactions described in Section 4.2. It does not introduce any information not covered in that section, only provides a visual aid to readers.}
\end{figure}

Given a valid encrypted solution stored in \scc and a secret key \csk, \se decrypts the ciphertext to retrieve a plaintext solution.
Figure~\ref{fig:ED-TLP-workflow} outlines the interaction between the parties in \dgc.  Below, we present \dgc in detail.
\begin{enumerate}[leftmargin=.56cm]
\item\brokenuline{$\cl.\setup{\mtdgm}(1^\lambda,\intervalvec)\rightarrow (\cpk, \csk)$}
\label{DGC-TLP::step::Client-side-Setup}.
Involves \cl.
    \begin{enumerate}[leftmargin=4mm]
    \item generate a secret key for symmetric-key encryption: \\
    $\skegen(1^\lambda)\rightarrow \csk$.
    \item set $\cpk=\intervalvec$ and send $(\cpk, \csk)$ to \se.
    \end{enumerate}

\item\brokenuline{$\cl.\delegate{\mtdgm}(\messvec, \csk)\rightarrow (\messvec^*, t_0)$}
\label{DGC-TLP::step::Client-side-Delegation}.
Involves \cl.
    \begin{enumerate}[leftmargin=4mm]
    \item encrypt each plaintext solution in vector \messvec:\\
    $\skeenc(\csk, m_i)\rightarrow m_i^*, \forall i \in [z]$.
    
    \item send $t_0$ and $\messvec^*=[{m}_1^*,.., m_z^*]$ to \tpc and $t_0$ to \se.
    \end{enumerate}

\item\brokenuline{$\se.\delegate{\mtdgm}(\ced$, $ToC$, $S$, \intervalvec, $aux$, $\adr_\tps$, $t_0$, $\Upsilon$, $\coinsvec)$ $\rightarrow \adr_\scc$}.
\label{DGC-TLP::step::Server-side-Delegation}
Involves \se.
\begin{enumerate}[leftmargin=4mm]

\item determine the extra delay $\Psi_j$ that \tps needs to find the $j$-th solution, $\forall j \in [z]$, by calling $\ced(ToC, S, \bar\Delta_j, aux)\rightarrow \Psi_j$.

\item construct a smart contract \scc, deploy \scc into the block\-chain, and deposit $\coins=\sum_{j=1}^z\coins_j$ amount of coins into \scc, where $\coins_j\in \coinsvec$.
Let $\adr_\scc$ be the address of the deployed \scc.

\item set the delivery time for the $j$-th solution as
$T_j=t_0 + \sum_{i=1}^j(\bar{\Delta}_i+\Psi_i +\Upsilon)$, $\forall j \in [z]$.

\item register $\timeexpvec=[T_1$, $\ldots$, $T_z]$ and $\adr_\tps$ in \scc.

\item send $\adr_\scc$ to \tpc and \tps.

\end{enumerate}

\item\brokenuline{$\tpc.\setup{\mtdgm}(1^\lambda, \intervalvec, S, z)\rightarrow (\pk,\sk)$}.
\label{DGC-TLP::step::Helper-side-Setup}
Involves \tpc.

\begin{enumerate}[leftmargin=4mm]
\item[$\bullet$] call $\setup{\gm}(1^\lambda, \intervalvec, S, z)\rightarrow (\pk,\sk)$, to generate public and private parameters for $z$ puzzles.
\end{enumerate}

\item\brokenuline{$\gen{\mtdgm}({\messvec^*}, \pk, \sk, t_0)\rightarrow \hat p$}.
\label{DGC-TLP::Generate-Puzzle-}
Involves \tpc.

\begin{enumerate}[leftmargin=4mm]
\item call $\gen{\mtdgm}( {\messvec^*}, \pk, \sk)\rightarrow \hat p$, to generate $z$ puzzles and their commitments.
Recall, $\hat p\coloneqq(\puzzvec, \commvec)$, where \puzzvec is a vector of puzzles and \commvec is a vector of commitments.

\item at time $t_0$, send \puzzvec to \tps and \commvec to \scc.

\end{enumerate}

\item\brokenuline{$\solve{\mtdgm}(\extravec, aux, \pk, \puzzvec, \adr_\scc, \coinsvec', \coinsvec) \rightarrow (\solsvec, q)$}.
\label{DGC-TLP::step::Solve-Puzzle}
Involves \tps.
\begin{enumerate}[leftmargin=4mm]
\item check if the deposit is sufficient ($\coins_j\geq\coins_j'$) and elements of \extravec are large enough. If not, set $q=0$ and halt.

\item call $\solve{\gm}(\pk, \puzzvec)\rightarrow (\solsvec, \messvec^*)$, to solve $z$ puzzles.

\end{enumerate}

\item \brokenuline{$\prove{\mtdgm}(\pk,s_j)\rightarrow \pi_j$}.
\label{prove-D}
Involves \tps.

\begin{itemize}[leftmargin=4mm]
\item[$\bullet$] call $\prove{\gm}(\pk, s_j)\rightarrow \pi_j$, to generate a proof, upon discovering a solution $s_j\in \solsvec$.
\end{itemize}

\item \brokenuline{$\register{\mtdgm}(s_j, \pi_j, \adr_\scc)\rightarrow t_j$}.
\label{DGC-TLP::step::Register-Puzzle}
Involves \tps.

\begin{itemize}[leftmargin=4mm]
    \item[$\bullet$] send the $j$-th encoded message $m_j^*$ (where $m_j^*\in s_j$) and its proof $\pi_j$ to \scc, to be registered in \scc by time $t_j$.
\end{itemize}

\item \brokenuline{$\verify{\mtdgm}(\pk, j, m_j^*, \pi_j, g_j, \Psi_j, t_j, t_0, \intervalvec, \Upsilon)\rightarrow \ver_j$}.
\label{verify-D}
Involves \scc.

\begin{enumerate}[leftmargin=4mm]
\item if this is the first time it is invoked (when $j=1$), set two vectors $\vervec=[\ver_1$, $\ldots$, $\ver_z]$ and $\paidvec=[\paid_1$, $\ldots$, $\paid_z]$ whose elements are initially set to empty $\empt$.

\item\label{dgmtlp::verify-delivery-time} read the content of \scc and check whether $(s_j$, $\pi_j)$ was delivered to \scc on time: $t_j\leq T_j=t_0 + \sum_{i=1}^j(\bar{\Delta}_i+\Psi_i + \Upsilon)$
\item\label{dgmtlp::verify-proof} call $\verify{\gm}(\pk, j,m_j^*, \pi_j, g_j)\rightarrow \ver'_j\in \bin$, to check a proof's validity, where $g_j\in \commvec$.

\item set $\ver_j=1$, if both checks in steps~\ref{dgmtlp::verify-delivery-time} and~\ref{dgmtlp::verify-proof} pass; set $\ver_j=0$, otherwise.

\end{enumerate}

\item \brokenuline{$\pay{\mtdgm}(\ver_j, \adr_\tps, \coinsvec, j)\rightarrow \paid_j$}.
\label{DGC-TLP::step::pay}
Involves \scc.
Invoked either by \cl or \se.

\begin{itemize}[leftmargin=4mm]
\item[$\bullet$] if $\ver_j=1$, then:
    \begin{enumerate}
    \item if $\paid_j\neq 1$, for $j$-th puzzle, send $\coins_j$ coins to \tps, where $\coins_j\in \coinsvec$.
    \item set $\paid_j=1$ to ensure \tps will not be paid multiple times for the same solution.
\end{enumerate}

\item[$\bullet$] if $\ver_j=0$, send $\coins_j$ coins back to $\se$ and set $\paid_j=0$.
\end{itemize}

\item \brokenuline{$\retrieve{\mtdgm}(\pk, \csk, m_j^*)\rightarrow m_j$}.
\label{DGC-TLP::step::retrive}
Involves \se.

\begin{enumerate}[leftmargin=4mm]
    \item read $m_j^*$ from \scc.
    \item locally decrypt $m^*_j$: $\skedec(\csk, m^*_j)\rightarrow m_j$.
\end{enumerate}

\end{enumerate}

In the \dgc (and the definition of \dgm) we assumed \tps is a rational adversary, to ensure only the timely delivery of the solution through our incentive mechanism.
Nevertheless, it is crucial to note that the privacy of the scheme remains intact even in the presence of a fully malicious \tps.

\begin{theorem}\label{theorem::D-TLP-sec}
If the symmetric-key encryption meets IND-CPA, GC-TLP is secure (w.r.t.\ Definition~\ref{def::GC-TLP-security}), the blockchain is secure (i.e., it meets {persistence} and {liveness} properties~\cite{GarayKL15}, and the underlying signature satisfies ``existential unforgeability under chosen message attacks'') and \scc's correctness holds, then \dgc is secure, w.r.t.\ Definition~\ref{def::DGM-TLP-security}.

\end{theorem}

 \subsection{Security Analysis of \dgc}\label{sec::D-TLP-Proof}

In this section, we prove the security of \dgc, i.e., Theorem~\ref{theorem::D-TLP-sec}.

\begin{proof}
We first focus on the solutions' privacy, w.r.t.\ Definition~\ref{Def::DGM-TLP::Solution-Privacy}.

\begin{claim}
If the symmetric-key encryption satisfies IND-CPA, then \dgc preserves solutions' privacy from \tpc and \tps, w.r.t.\ Case~\ref{def::solution-privacy-from-helper} in Definition~\ref{Def::DGM-TLP::Solution-Privacy}.
\end{claim}
\begin{proof}
\tpc receives a vector $\messvec^*$ of ciphertexts (i.e., encrypted plaintext solutions), and public parameters in set $A=\{\cpk$, $\ced$, $ToC$, $S$, \intervalvec, $aux$, $\coins$, $t_0$, $\adr_\tps$, $\adr_\scc\}$.
Other parameters given to $\adv_1$ in the experiment (i.e., parameters in set $B=\{\pk, \solsvec, \hat p\}$) are generated by \tpc itself and may seem redundant.
However, we have given these parameters to $\adv_1$ for the case where \tps is corrupt as well, which we will discuss shortly.

Since the public parameters were generated independently of the plaintext solutions $m_1, \ldots, m_z$, they do not reveal anything about each $m_i$.
As the symmetric-key encryption scheme is IND-CPA, the vector $\messvec^*$ of ciphertext reveals no information about each $m_i$.
Specifically, in Case~\ref{def::solution-privacy-from-helper} in Definition~\ref{Def::DGM-TLP::Solution-Privacy}, the probability that $\adv_1$ can tell whether a ciphertext $m^*_{\rb_{i},i}\in \messvec^*$ is an encryption of message $m_{0,i}$ or $m_{1,i}$, chosen by $\adv_1$, is at most $\frac{1}{2}+\negl$.

Now, we focus on the case where \tps is corrupt. 
The messages that \tps receives include a vector $\messvec^{*}$ of ciphertexts and parameters in set $A+B$.
As long as the symmetric-key encryption meets IND-CPA, the vector $\messvec^*$ of ciphertext reveals no information about each $m_i$.
Hence, when \tps is corrupt, the probability that $\adv_1$, in Case~\ref{def::solution-privacy-from-helper} in Definition~\ref{Def::DGM-TLP::Solution-Privacy}, can tell if a ciphertext $m^*_{\rb_{i},i}\in \messvec^*$ is the encryption of message $m_{0,i}$ or $m_{1,i}$, both of which were initially chosen by $\adv_1$, is at most $\frac{1}{2}+\negl$.

Recall the public parameters in $A$ do not reveal anything about each $m_i$.
The same holds for the public parameter $\pk\in B$.
Also, parameters \solsvec and $\hat{p}$ in $B$ were generated by running  $\gen{\mtdgm}$ and $\solve{\mtdgm}$ on the ciphertexts in vector $\messvec^*$.
Thus, \solsvec and $\hat{p}$ will not reveal anything about the plaintext messages, as long as the symmetric-key encryption satisfies IND-CPA\@.
\end{proof}
\begin{claim}
If \gc protocol is privacy-preserving (w.r.t.\ Definition~\ref{Def::GM-TLP::Solution-Privacy}), then \dgc preserves solutions' privacy from \se and \tps, w.r.t.\ Case~\ref{def::solution-privacy-from-solver} in Definition~\ref{Def::DGM-TLP::Solution-Privacy}.
\end{claim}
\begin{proof}
Before solving the puzzles, the messages that \tps or \se receives include the elements of sets $A=\{\pk, \hat p\}$ and $B=\{\extravec$, $\ced$, $t_0$, $\coins$, $\adr_\tps$, $\adr_\scc$, $ToC$, $aux\}$.
The elements of set $A$ are identical to what \se receives in \gc.
The elements of set $B$ are independent of the plaintext solutions and the puzzles' secret and public parameters.
Therefore, knowledge of $B$ does not help \se or \tps learn the plaintext solutions before solving the puzzles.
More formally, given $A+B$, in Case~\ref{def::solution-privacy-from-solver} in Definition~\ref{Def::DGM-TLP::Solution-Privacy}, the probability that $\adv_3$ can tell whether a puzzle $p_{\rb_{i},i}\in \puzzvec\in \hat p$ has been created for plaintext message $m_{0,i}$ or $m_{1,i}$ is at most $\frac{1}{2}+\negl$, due to the privacy property of \gc, i.e., Theorem~\ref{Solution-Privacy}.
\end{proof}
\begin{claim}
If \gc meets solution-validity (w.r.t.\ Definition~\ref{Def::GM-TLP::Solution-Validity}), the blockchain is secure (i.e., it meets {persistence} and {liveness} properties~\cite{GarayKL15}, and the signature satisfies existential unforgeability under chosen message attacks) and the smart contract's correctness holds, \dgc preserves a solution validity, w.r.t.\ Definition~\ref{Def::DGM-TLP::Solution-Validity}.
\end{claim}
\begin{proof}
First, we focus on event $\text{I}=(con_1 \wedge con_2\wedge con_3)$: a prover submits proof on time and passes the verification despite the proof containing an opening for a different message than the one already committed to.
Compared to \gc, the extra information that a prover (in this case \tps) learns in \dgc includes parameters in set $A+B$.
Nevertheless, these parameters are independent of the plaintext messages and the parameters used for the commitment.
Therefore, the solutions' validity is reduced to the solution validity of \gc and, in turn, to the security of the commitment scheme.
Specifically, given commitment $g_j=\hash(m_j,d_j)$ and an opening $\pi\coloneqq(m_j,d_j)$, for an adversary to break the solution validity, it must generate $(m'_j, d'_j)$, such that $\hash(m'_j, d'_j)=g_j$, where $m_j\neq m'_j$.
However, this is infeasible for a PPT adversary, as $\hash(\cdot)$ is collision-resistant, in the random oracle model.
Thus, in the experiment in Definition~\ref{Def::DGM-TLP::Solution-Validity}, event \text{I} occurs with a probability at most $\negl$.

Now, we focus on event $\text{II}=(\neg con_1 \wedge con_4 \wedge con_5)$: the adversary has generated a valid proof and passed the verification despite registering the proof late.
Due to the persistency property of the blockchain, once a transaction goes more than $v$ blocks deep into the blockchain of one honest player (where $v$ is a security parameter), it will be included in every honest player's blockchain with overwhelming probability, and it will be assigned a permanent position in the blockchain (so it will not be modified with an overwhelming probability).

Due to the liveness property, all transactions originating from honest parties will eventually end up at a depth of more than $v$ blocks in an honest player's blockchain; so, the adversary cannot perform a selective denial of service attack against honest account holders.
Thus, with a high probability when a (well-formed transaction containing) proof is sent late to the smart contract, the smart contract declares late; accordingly, $\verify{\mtdgm}(\cdot)$ outputs $0$ except for a negligible probability, \negl.
Hence, in the experiment in Definition~\ref{Def::DGM-TLP::Solution-Validity}, event \text{II} occurs with a probability at most \negl.
\end{proof}

Now, we focus on fair payment, w.r.t.\ Definition~\ref{def::dgm-tlp-fair-payment}.

\begin{claim}
If the blockchain is secure (it meets persistence and liveness properties~\cite{GarayKL15}), then \dgc offers fair payment, w.r.t.\ Definition~\ref{def::dgm-tlp-fair-payment}.
\end{claim}
\begin{proof}
The proof reduces to the security of the blockchain (and smart contracts).
Specifically, due to the persistence and liveness properties,
\begin{enumerate*}[label=(\arabic*)]
    \item the state of a smart contract cannot be tampered with and
    \item each function implemented in a smart contract correctly computes a result, except for a negligible probability \negl.
\end{enumerate*}
Thus, in the experiment in Definition~\ref{def::dgm-tlp-fair-payment}, when $\verify{\mtdgm}(\pk$, $j$, $\pi_j$, $g_j$, $\Psi_j$, $t_j$, $t_0$, $\intervalvec$, $\Upsilon)\rightarrow \ver_j\in \bin$, then
\begin{enumerate*}[label=(\arabic*)]
    \item the intact $\ver_j$ is passed on to $\pay{\mtdgm}(\ver_j$, $\adr_\tps$, \coinsvec, $j)$ as input (because the smart contract generates and maintains $\ver_j$ and passes it to $\verify{\mtdgm}$) and
    \item $\paid_j=\ver_j$, except for the probability of \negl.
\end{enumerate*}
\end{proof}

\renewcommand{\qedsymbol}{$\blacksquare$}
Hence, \dgc is secure, w.r.t.\ Definition~\ref{def::DGM-TLP-security}, given that \dgc satisfies the solutions' privacy (w.r.t.\ Definition~\ref{Def::DGM-TLP::Solution-Privacy}), the solutions' validity (w.r.t.\ Definition~\ref{Def::DGM-TLP::Solution-Validity}), and fair payment (w.r.t.\ Definition~\ref{def::dgm-tlp-fair-payment}).
\end{proof} 
 

\subsection{Satisfying the Primary Features}
Let us explain how \dgc satisfies the primary properties.

\begin{itemize}[leftmargin=4.2mm]
\item[$\bullet$] \textit{Varied Size Time Intervals and Efficiently Handling Multiple Puzzles.}
\dgc uses \gc in a block-box manner, therefore, it inherits these features.

\item[$\bullet$] \textit{Privacy.}
To ensure plaintext solutions' privacy, \cl encrypts all plaintext solutions using symmetric-key encryption and asks \tpc to treat the ciphertexts as puzzles' solutions.
\cl sends the secret key of the encryption only to \se.
This approach protects the privacy of plaintext messages from both \tpc and \tps.

\item[$\bullet$]\textit{Exact-time Solutions Recovery.}
In \dgc, given the exact computational power of \tps, \se uses $\ced(\cdot)$ to determine the extra time \tps needs for each solution.
Thus, instead of assuming that \tps possesses computing resources to perform the maximum number of squarings $S$ per second and find a solution on time, \se considers the available resources of \tps.

\item[$\bullet$] \textit{Timely Delivery of Solutions and Fair Payment.}
To ensure the timely \emph{delivery} of a solution, \se constructs a smart contract \scc and specifies by when each solution must be delivered, based on $\ced(\cdot)$ output, and deposits in \scc a certain amount of coins.
\se requires \scc to check if each (encrypted) solution is valid and delivered on time and, if the two checks pass, pay \tps.
This gives \tps the assurance it will get paid if it provides a valid solution on time.
The hash-based \scc-side verification imposes a low computation cost.
\end{itemize}

\section{Evaluation}\label{sec:evaluation}
We conducted an in-depth evaluation of \gc and \dgc against
\begin{enumerate*}[label=(\arabic*)]
    \item  the original RSA-based TLP of~\citeauthor{Rivest:1996:TPT:888615}~\cite{Rivest:1996:TPT:888615} because it serves as the foundation for numerous TLPs and VDFs, and
    \item the \ctlp of~\citeauthor{Abadi-C-TLP}~\cite{Abadi-C-TLP} because it is the most efficient multi-instance puzzle that supports efficient verification.
\end{enumerate*}
We analyze the features of all protocols and compare their overheads asymptotically and concretely.

\subsection{Features}\label{subsec:features}

Table~\ref{table::comparisonTable} compares the four schemes' features. Figure~\ref{fig:TLP-schemes} further illustrates differences among these four schemes concerning the support for multi-puzzle and varied-size time intervals.

\begin{table}[h]
\setlength{\smalltabcolsep}{0.7\tabcolsep}
\caption[Differences between the four schemes in terms of features]{Differences in features between:
    \begin{enumerate*}[label=(\arabic*)]
        \item the original TLP of~\citeauthor{Rivest:1996:TPT:888615}~\cite{Rivest:1996:TPT:888615},
        \item the \ctlp of~\citeauthor{Abadi-C-TLP}~\cite{Abadi-C-TLP}
        \item our \gc, and
        \item our \dgc.
    \end{enumerate*}
    }\label{table::comparisonTable} 
\begin{tabularx}{\linewidth}{
X@{}
c@{\hspace{\smalltabcolsep}}c@{\hspace{\smalltabcolsep}}c@{\hspace{\smalltabcolsep}}
c@{}
} 
\toprule
&\multicolumn{4}{c}{Scheme}\\
\cmidrule(lr){2-5}
Features & TLP & C-TLP & \gc & \dgc \\
\cmidrule(lr){1-1} \cmidrule(lr){2-5}
Multi-puzzle & \xmark & \cmark & \cmark & \cmark\\
Varied-size & \cmark & \xmark & \cmark & \cmark\\
Verification & \xmark & \cmark & \cmark & \cmark\\
Delegation & \xmark & \xmark & \xmark & \cmark\\
Exact-time Solution Recovery & \xmark & \xmark & \xmark & \cmark\\
Fair Payment & \xmark & \xmark & \xmark & \cmark\\
\bottomrule
\end{tabularx}
\end{table}

\begin{figure}[!b]
    \centering
    \includegraphics[width=\linewidth]{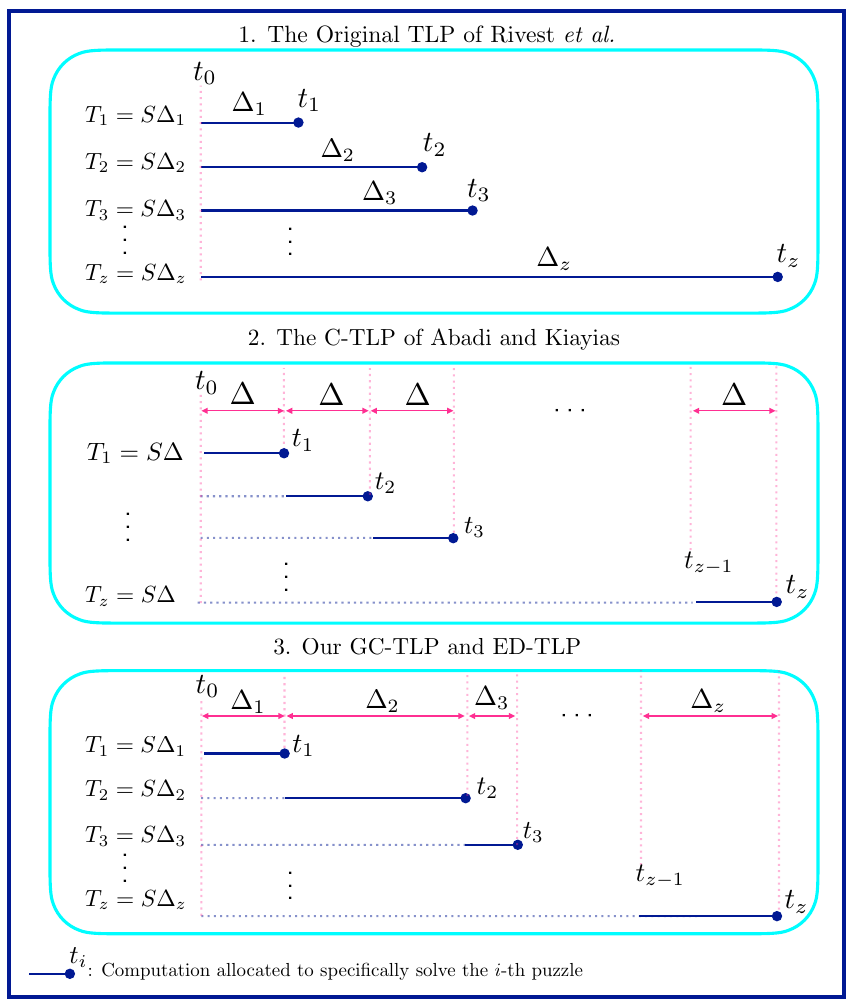}
    \caption{Differences between the following four schemes in terms of supporting multi-puzzle and varied-size time intervals: the original TLP of~\citeauthor{Rivest:1996:TPT:888615}~\cite{Rivest:1996:TPT:888615}, the C-TLP\xspace of~\citeauthor{Abadi-C-TLP}~\cite{Abadi-C-TLP},
    our GC-TLP\xspace, and
    our ED-TLP\xspace.
    }
    \Description{The figure shows that the original Time-Lock Puzzle by Rivest et al. requires each puzzle in a sequence to be solved in parallel for its full duration; the Chained Time-Lock Puzzle of Abadi and Kiayias allows the puzzles to be solved sequentially but only if they are of equal intervals. Our Efficient Delegated Time-Lock Puzzle allows puzzles with arbitrary intervals to be solved sequentially, therefore has high flexibility and low overhead.}
    \label{fig:TLP-schemes}
\end{figure}

The features of our protocols were motivated by limitations in the existing literature, illustrated by the multi-client example in Section~\ref{sec:introduction}.
Our scheme enables clients with different puzzles to combine them in a single chain the server can solve sequentially. 
We briefly explain how it can be done in a simple two-client case. 
Consider clients $\cl_1$ and $\cl_2$ who encode messages $m_1$ and $m_2$ in puzzles $\hat{p}_1$ and $\hat{p}_2$ with intervals $\Delta_1$ and $\Delta_2$, respectively, then send them to a server \se at the same time.
Assume $\Delta_1<\Delta_2$.
After \se receives $\hat{p}_1$ and $\hat{p}_2$, it can start working on $\hat{p}_1$ immediately, but also asks (or incentivizes) the clients to chain their messages.
To do this, $\cl_2$, with a longer interval, delegates its messages to $\cl_1$.
Acting as \tpc, $\cl_1$ uses its existing keys to extend $\hat{p}_1$ with a puzzle containing $m_2$ set for interval of $\bar{\Delta}_2 = \Delta_2 - \Delta_1$, which is sent to \se.
After solving the initial puzzle in $\hat{p}_1$, \se can start solving the new puzzle.
It is thus able to recover both $m_1$ and $m_2$ in time $\Delta_2$ using only a single CPU, doing equivalent work to the work $\hat{p}_2$ initially required.
This \textit{scales to any number of clients}, as puzzles chains can be arbitrarily long.
Alternatively, clients can delegate puzzle generation to a common helper \tpc from the start.
In addition, if \se is unable to solve the puzzle chain in useful time, it can use server delegation to securely offload its work, with fair payment and timely delivery of solutions

The protocols are modular, so the main features of \dgc: variable intervals, client delegation, and server delegation can be used separately.
For example, if $\cl_1$ and $\cl_2$ are offline after sending their puzzles and cannot chain them, \se can still delegate the work and obtain the solutions on time (at a higher cost).

\subsection{Asymptotic Cost Analysis}\label{subsec:asymptotic-cost-analysis}

We provide a brief asymptotic cost analysis of the protocol.
Table~\ref{table::puzzle-com} summarizes the cost comparison between the four schemes.
In our cost evaluation, for the sake of simplicity, we do not include the output of $\ced(\cdot)$, as the output is a fixed value and remains the same for all the schemes we analyze in this section.

\begin{table*}[!t]
\setlength{\smalltabcolsep}{0.7\tabcolsep}
\caption{
    Asymptotic costs comparison.
    In the table, $z$ is the total number of puzzles.
    $\Delta$ and $\Delta_i$ are time intervals in C-TLP and TLP respectively.
    $\bar{\Delta}_i$ is a time interval in GC-TLP\xspace and ED-TLP\xspace with $\Delta_i = \sum_{j=1}^i{\bar{\Delta}_j}$.
}
\vspace{-0.8em}
\label{table::puzzle-com}
\begin{tabular}{
@{}cp{2cm} 
c@{\hspace{\smalltabcolsep}}c 
c@{\hspace{\smalltabcolsep}}c 
c@{\hspace{\smalltabcolsep}}c 
c@{}c 
cc 
c@{\hspace{\smalltabcolsep}}c@{} 
} 
\toprule

&
&\multicolumn{10}{c}{{Algorithms Complexity}}
&\multicolumn{2}{c}{{Total Complexity}} \\
\cmidrule(rl){3-12}\cmidrule(l){13-14}

&
&\multicolumn{2}{c}{\ensuremath{\mathsf{Setup}}}
&\multicolumn{2}{c}{\ensuremath{\mathsf{Delegate}}}
&\multicolumn{2}{c}{\ensuremath{\mathsf{GenPuzzle}}}
&\multicolumn{2}{c}{\ensuremath{\mathsf{Solve}}}
&\multirow{3}{*}{\ensuremath{\mathsf{Verify}}}
&\multirow{3}{*}{\ensuremath{\mathsf{Retrieve}}}
&\multirow{3}{*}{Comp.}
&\multirow{3}{*}{Comm.} \\
\cmidrule(rl){3-4}\cmidrule(rl){5-6}\cmidrule(rl){7-8}\cmidrule(rl){9-10}

{\multirow{-4}{*}{\rotatebox[origin=c]{90}{\textbf{Scheme}}}}
&\multirow{-4}{*}{Operation}
&\ensuremath{\mathcal{C}}
&\ensuremath{\mathcal{HC}}
&\ensuremath{\mathcal{C}}
&\ensuremath{\mathcal{S}}\
&\ensuremath{\mathcal{C}}
&\ensuremath{\mathcal{HC}}
&\ensuremath{\mathcal{S}}
&\ensuremath{\mathcal{HS}}
&
&
&
& \\
\cmidrule(r){1-2}\cmidrule(rl){3-4}\cmidrule(rl){5-6}\cmidrule(rl){7-8}\cmidrule(rl){9-10}\cmidrule(rl){11-11}\cmidrule(rl){12-12}\cmidrule(l){13-14}

& Exp.
&$-$
&$O(z)$
&$-$
&$-$
&$-$
&$O(z)$
& $-$
&{{ $O(S\cdot\sum\limits^{z}_{i=1}\bar{\Delta}_{i} )$}}
& $-$
& $-$
&
& \\
    
& Add./Mul.
&$-$
&$-$
&$-$
&$O(z)$
&$-$
&$O(z)$
& $-$
& $O(z)$
&  $-$
& $-$
&
& \\

& Hash
&$-$
&$-$
&$-$
&$-$
&$-$
&$O(z)$
&$-$
&$-$
&$O(z)$
&$-$
&
& \\

\multirow{-5}{*}{\rotatebox[origin=c]{90}{\textbf{ED-TLP\xspace}}}
& Sym.\ Enc
&$-$
&$-$
&$O(z)$
&$-$
&$O(z)$
&$O(z)$
&$-$
&$O(z)$
&$-$
&$O(z)$
&\multirow{-5}{*}{$O(S\cdot\sum\limits^{z}_{i=1}\bar{\Delta}_{i} )$}
&\multirow{-5}{*}{$O(z)$} \\
\cmidrule(r){1-2}\cmidrule(rl){3-4}\cmidrule(rl){5-6}\cmidrule(rl){7-8}\cmidrule(rl){9-10}\cmidrule(rl){11-11}\cmidrule(rl){12-12}\cmidrule(l){13-14}
 
& Exp.
&$O(z)$
&
&
&
&$O(z)$
&
&$O(S\cdot\sum\limits^z_{i=1}\bar{\Delta}_{i})$
&
&$-$
&
&
& \\

& Add./Mul.
&$-$
&
&
&
&$O(z)$
&
&$O(z)$
&
&$-$
&
&
& \\

& Hash
&$-$
&
&
&
&$O(z)$
&
&$-$
&
&$O(z)$
&
&
& \\

\multirow{-5}{*}{\rotatebox[origin=c]{90}{ \textbf{GC-TLP\xspace}}}
& Sym. Enc
&$-$
&
&
&
&$O(z)$
&
&$O(z)$
&
&$-$
&
&\multirow{-5}{*}{$O(S\cdot\sum\limits^z_{i=1}\bar{\Delta}_i)$}
&\multirow{-5}{*}{$O(z)$} \\
\cmidrule(r){1-2}\cmidrule(rl){3-4}\cmidrule(rl){5-6}\cmidrule(rl){7-8}\cmidrule(rl){9-10}\cmidrule(rl){11-11}\cmidrule(rl){12-12}\cmidrule(l){13-14}

& Exp.
&$1$
&
&
&
&$O(z)$
&
&$O(z\cdot S\cdot \Delta)$
&
&$-$
&
&
& \\

& Add./Mul.
&$-$
&
&
&
&$O(z)$
&
&$O(z)
$&
&$-$
&
&
& \\

& Hash
&$-$
&
&
&
&$O(z)$
&
&$-$
&
&$O(z)$
&
&
& \\

\multirow{-4}{*}{\rotatebox[origin=c]{90}{ \textbf{C-TLP}}}
& Sym. Enc
&$-$
&
&
&
&$O(z)$
&
&$O(z)$
&
& $-$
&
&\multirow{-3}{*}{$O(z\cdot S\cdot \Delta)$}
&\multirow{-3}{*}{$O(z)$} \\
\cmidrule(r){1-2}\cmidrule(rl){3-4}\cmidrule(rl){5-6}\cmidrule(rl){7-8}\cmidrule(rl){9-10}\cmidrule(rl){11-11}\cmidrule(rl){12-12}\cmidrule(l){13-14}

& Exp.
&$O(z)$
&
&
&
&$O(z)$
&
&$O(S\cdot\sum\limits^z_{i=1}\Delta_i)$
&
&
&
&
& \\ 

& Add./Mul.
&$-$
&
&
&
&$O(z)$
&
&$O(z)$
&
&
&
&
& \\

& Hash
&$-$
&
&
&
&$-$
&
&$-$
&
&
&
&
& \\

\multirow{-5}{*}{\rotatebox[origin=c]{90}{\textbf{TLP}}}
& Sym.\ Enc
&$-$
&
&
&
&$O(z)$
&
&$O(z)$
&
&
&
&\multirow{-5}{*}{$O(S\cdot\sum\limits^z_{i=1}\Delta_i)$}
&\multirow{-5}{*}{$O(z)$} \\
\bottomrule
\end{tabular}
\end{table*}

\subsubsection{Computation Cost}
As Table~\ref{table::puzzle-com} demonstrates, \dgc has the lowest client-side setup and puzzle generation cost and negligible computation cost for the server.
Schemes that support verification: C-TLP, \gc, and \dgc, all have $O(z)$ verification cost.

Thus, in \dgc, the client does not need to perform any modular exponentiation.
However, in the rest of the schemes, the client has to perform $O(z)$ modular exponentiation in the setup and/or puzzle generation phases.
Furthermore, in \dgc, the server does not engage in any modular exponentiation, in contrast to other schemes that require the server's involvement in performing $O(S\cdot\sum^z_{i=1}\bar{\Delta}_i )$ exponentiation in \gc, $O(z\cdot S\cdot \Delta)$ exponentiation within C-TLP, and $O(S\cdot\sum^z_{i=1}{\Delta}_i )$ exponentiation in the original TLP.

\subsubsection{\dgc's Added Cost}
\dgc imposes an additional cost, $O(z)$, during the retrieve phase, distinguishing it from the other schemes that do not incur such a cost. 
Note that the retrieve phase exclusively involves invocations of the symmetric key encryption scheme, which imposes low computational costs.



\subsubsection{Communication Cost}
As depicted in Table~\ref{table::puzzle-com}, for the $z$-puzzle setting, the communication complexity for all four schemes is $O(z)$.
However, in \dgc, both \cl and \se only transmit random values and outputs of symmetric-key encryption, which are shorter than the messages in the other three schemes.
Thus, \dgc exhibits the lowest client-side and server-side concrete communication costs.

\subsection{Concrete Run-time Analysis}\label{subsec:concrete-run-time-analysis}

\begin{table*}[!ht]
\setlength{\smalltabcolsep}{6pt}

\sisetup{
table-format=5.2,
round-mode = places,
round-precision = 2,
print-zero-integer=true,
table-text-alignment=right,
table-align-text-before=false,
group-digits = none,
}
\caption{Run time using 2048-bit moduli and 1 second intervals (or equivalent cumulative durations in TLP).}\label{table::puzzle-time} 
\vspace{-0.8em}
\begin{tabular}{
c@{\hspace{5pt}} 
S[table-format=5.0, round-precision = 0]@{} 
S[table-format=\ensuremath{<}3.2]@{\hspace{\smalltabcolsep}} 
S[table-format=1.2]@{} 
S[table-format=3.2]@{\hspace{\smalltabcolsep}} 
S[table-format=3.2]@{} 
S[table-format=2.2]@{\hspace{\smalltabcolsep}} 
S[table-format=3.2]@{\hspace{\smalltabcolsep}} 
S[table-format=5.2]@{} 
S[table-format=5.2]@{\hspace{\smalltabcolsep}} 
S[table-format=3.2]@{\hspace{\smalltabcolsep}}
S[table-format=3.2]@{}
>{\columncolor{white}[-2pt][2pt]}S[table-format=3.2]@{\hspace{\smalltabcolsep}} 
S[table-format=3.2]@{} 
>{\columncolor{white}[-11pt][2pt]}S[table-format=5.2]@{\hspace{\smalltabcolsep}} 
S[table-format=5.2]@{\hspace{\smalltabcolsep}} 
S[table-format=5.2]@{} 
}
\toprule
 
&
&\multicolumn{15}{c}{{ Algorithms Run Time (in seconds) }} \\\cmidrule(l){3-17}

&
&\multicolumn{2}{c}{\ensuremath{\mathsf{Setup}}}
&\multicolumn{2}{c}{\ensuremath{\mathsf{Delegate}}}
&\multicolumn{2}{c}{\ensuremath{\mathsf{GenPuzzle}}}
&\multicolumn{2}{c}{\ensuremath{\mathsf{Solve}}}
&
&
&\multicolumn{5}{c}{Total time} \\

\cmidrule(rl){3-4}\cmidrule(rl){5-6}\cmidrule(rl){7-8}\cmidrule(rl){9-10}\cmidrule(l){13-17}
   
{\multirow{-4}{*}{{\rotatebox[origin=c]{90}{\textbf{Scheme}}}}}

&{\multirow{-4}{*}{\parbox{1cm}{\centering Number\\of\\puzzles}}}
&{\ensuremath{\mathcal{C}}}
&{\ensuremath{\mathcal{HC}}}
&{\ensuremath{\mathcal{C}}}
&{\ensuremath{\mathcal{S}}}
&{\ensuremath{\mathcal{C}}}
&{\ensuremath{\mathcal{HC}}}
&{\ensuremath{\mathcal{S}}}
&{\ensuremath{\mathcal{HS}}}
&{\multirow{-2}{*}{\ensuremath{\mathsf{Verify}}}}
&{\multirow{-2}{*}{\ensuremath{\mathsf{Retrieve}}}}
&{\ensuremath{\mathcal{C}}}
&{\ensuremath{\mathcal{HC}}}
&{\ensuremath{\mathcal{S}}}
&{\ensuremath{\mathcal{HS}}}
&{Overall}\\ 

\cmidrule(rl){1-2}\cmidrule(rl){3-4}\cmidrule(rl){5-6}\cmidrule(rl){7-8}\cmidrule(rl){9-10}\cmidrule(rl){11-11}\cmidrule(rl){12-12}\cmidrule(l){13-17}

&10
& \ensuremath{<} 0.01
&0.03
& \ensuremath{<}0.01
& \ensuremath{<}0.01
&
&0.028356042
&
&10.66745721
&\ensuremath{<}0.01
&\ensuremath{<}0.01
&{\cellcolor{T-Q-PH3}}\ensuremath{<}0.01
&0.06
&{\cellcolor{T-Q-PH3}}\ensuremath{<}0.01
&10.67
&10.72 \\

&100
&\ensuremath{<}0.01
&0.03
&\ensuremath{<}0.01
&\ensuremath{<}0.01
&
&0.19
&
&107.50
&\ensuremath{<}0.01
&\ensuremath{<}0.01
&{\cellcolor{T-Q-PH3}}\ensuremath{<}0.01
&0.22
&{\cellcolor{T-Q-PH3}}\ensuremath{<}0.01
&107.50
&107.72 \\
  
&1000
&\ensuremath{<}0.01
&0.05
&0.02
&\ensuremath{<}0.01
&
&1.80
&
&1071.01
&\ensuremath{<}0.01
&0.02
&{\cellcolor{T-Q-PH3}}0.02
&1.85
&{\cellcolor{T-Q-PH3}}0.02
&1071.01
&1072.88 \\

\multirow{-4}{*}{\rotatebox[origin=c]{90}{\textbf{ED-TLP\xspace}}}
&10000
&\ensuremath{<}0.01
&0.23
&0.13
&\ensuremath{<}0.01
&
&18.12
&
&10669.92
&0.03
&0.13
&{\cellcolor{T-Q-PH3}}0.13
&18.12
&{\cellcolor{T-Q-PH3}}0.13
&10669.92
&10688.54 \\ \cmidrule(rl){1-2}\cmidrule(rl){3-4}\cmidrule(rl){5-6}\cmidrule(rl){7-8}\cmidrule(rl){9-10}\cmidrule(rl){11-11}\cmidrule(rl){12-12}\cmidrule(l){13-17}

&10
&0.03
&
&
&
&0.03
&
&10.71
&
&\ensuremath{<}0.01
&
&0.05
&
&10.71
&
&10.76 \\

&100
&0.03
&
&
&
&0.19
&
&107.28
&
&\ensuremath{<}0.01
&
&0.22
&
&107.28
&
&107.50 \\

&1000
&0.05
&
&
&
&1.81
&
&1072.96
&
&0.02
&
&1.85
&
&1072.96
&
&1074.81 \\

\multirow{-4}{*}{\rotatebox[origin=c]{90}{\textbf{GC-TLP\xspace}}}
&10000
&0.22
&
&
&
&18.07
&
&10689.07
&
&0.02
&
&18.29
&
&10689.07
&
&10707.37 \\ \cmidrule(rl){1-2}\cmidrule(rl){3-4}\cmidrule(rl){5-6}\cmidrule(rl){7-8}\cmidrule(rl){9-10}\cmidrule(rl){11-11}\cmidrule(rl){12-12}\cmidrule(l){13-17}

&10
&0.03
&
&
&
&0.02
&
&10.70
&
&\ensuremath{<}0.01
&
&0.05
&
&10.70
&
&10.75 \\
    
&100
&0.03
&
&
&
&0.19
&
&107.09
&
&\ensuremath{<}0.01
&
&0.21
&
&107.01
&
&107.22 \\

&1000
&0.04
&
&
&
&1.83
&
&1070.36
&
&0.03
&
&1.86
&
&1070.36
&
&1072.23 \\

\multirow{-4}{*}{\rotatebox[origin=c]{90}{\textbf{C-TLP} }}
&10000
&0.14
&
&
&
&18.10
&
&10680.29
&
&0.02
&
&18.24
&
&10680.29
&
&10698.54 \\ \cmidrule(rl){1-2}\cmidrule(rl){3-4}\cmidrule(rl){5-6}\cmidrule(rl){7-8}\cmidrule(rl){9-10}\cmidrule(rl){11-11}\cmidrule(rl){12-12}\cmidrule(l){13-17}

&10
&0.14
&
&
&
&0.02
&
&58.50
&
&
&
&0.16
&
&58.50
&
&58.65\\

&100
&1.86
&
&
&
&0.19
&
&5075.49
&
&
&
&2.05
&
&5075.49
&
&5077.53 \\
  
&1000
&19.78
&
&
&
&1.80
&
&N/A$^a$
&
&
&
&21.58
&
&N/A$^a$
&
&N/A$^a$ \\

\multirow{-4}{*}{\rotatebox[origin=c]{90}{\textbf{TLP}}}
&10000
&181.77
&
&
&
&17.15
&
&N/A$^b$
&
&
&
&198.92
&
&N/A$^b$
&
&N/A$^b$ \\ \bottomrule
\multicolumn{5}{l}{\footnotesize
$^a$ Estimated total runtime: 5.8 days
}&
\multicolumn{10}{l}{\footnotesize
$^b$ Estimated total runtime: 1.5 years
}\\
\end{tabular}
\end{table*}

To complement the asymptotic analysis, we implemented
\dgc and \gc, as well as \ctlp~\cite{Abadi-C-TLP} and the \citeauthor{Rivest:1996:TPT:888615} TLP~\cite{Rivest:1996:TPT:888615} in Python 3 using the high-performance multiple precision integer library GMP~\cite{GranlundGMP} via the gmpy2~\cite{gmpy2} library.
The implementation strictly adhere to the protocol descriptions for TLP, \ctlp, and \gc.
Smart contract functionality was implemented using Solidity.
For protocols requiring a commitment scheme, we used a SHA512 hash-based commitment with 128-bit witness values. 

\subsubsection{Setting the Parameters} 
\begin{figure}[t]
\includegraphics[width=\linewidth]{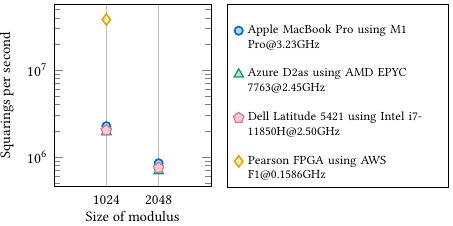}
\caption{Log plot of the squaring capabilities of consumer and commodity cloud hardware for 1024-bit and 2048-bit RSA groups, compared to an FPGA-based solution for 1024-bit moduli as reported in~\cite{FPGACompRes}. General purpose CPUs have similar capabilities, but the FGPA solution is an order of magnitude faster.}
\Description{The plot shows that the various general purpose CPUs are clustered around 2000000 operations per second for 1024-bit moduli and 800000 operations per second for 2048-bit moduli. FPGA solutions are capable of 38000000 operations per second.}
\label{fig:bench}
\end{figure}
The evaluation required a concrete value of $S$, the number of repeated squaring a CPU can perform.
We benchmarked 10 different commodity computing machines and consumer devices. 
The results are presented in \autoref{fig:bench}, 
 with the full results presented in~Table~\ref{table:bench-} in the Appendix~\ref{Squaring-Capabilities}.

\subsubsection{Run-time Comparison}
For the concrete analysis, we ran the four protocols on a 2021 MacBook Pro equipped with an Apple M1 Pro CPU. 
As we use execution time as a proxy for computation costs, for benchmarking we used a mock smart contract running locally, without an underlying blockchain or network delay.
Table~\ref{table::puzzle-time} lists the running time of each protocol step and totals per actor.

The most salient observation is \se can offload $100$\% of its workload to \tps with negligible computational overhead.
Although \dgc adds the \delegate{} and \retrieve{} steps, their CPU costs for \se are negligible per puzzle.
In our benchmarks, both \se and \tps were run on the same hardware, but as \autoref{fig:bench} shows, delegating to specialized hardware can result in large gains, especially for an underpowered server. 
Client-side delegation also has low overheads.
Despite its much lower computational complexity, delegating puzzle generation results in computational savings of $99$\% for the client.
As puzzle generation can be parallelized, generating large numbers of puzzles can benefit from delegation to a \tpc with a high number of powerful CPU cores.
For example, $10000$ puzzles only require $2.48$s on 8 cores, compared to $18.12$s on a single core.
To support variable durations, \gc has a more complex setup compared to \ctlp (Table~\ref{table::puzzle-time}).
In practice, the effect is minor per puzzle, with a total increase of $0.08$s for $10000$ puzzles.

\begin{figure}[!t]
\centering
\includegraphics[width=\linewidth]{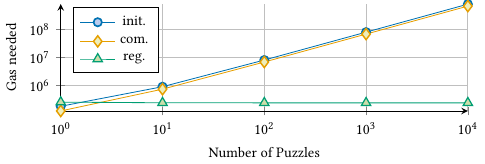}
\caption{Gas amount needed to execute the smart contract functions (initialization, committing to the puzzles, and registering solutions) for different puzzle instances}
\Description{The number of puzzles leads to a fixed registration cost and a linear increase in initialization and commitment costs.}
\label{fig:gas-cost}
\end{figure}
\subsection{\dgc Gas Cost Estimates}
We implemented the \dgc smart contract in Solidity and estimated its costs on the Ethereum~\cite{ethereum} Mainnet and two Ethereum-based level 2 networks: Arbitrum One~\cite{kalodner2018arbitrum} and Polygon PoS~\cite{bjelic2017pol}.
Although the protocol is agnostic of the underlying blockchain, we chose Solidity and Ethereum-based chains due to their prevalence. Figure~\ref{fig:gas-cost} shows the gas cost scaling with corresponding USD costs list in Table~\ref{tab:cost_comparison}.
Initializing the contract, executed when \se deploys the contract in $\se.\delegate{\mtdgm}$, and committing to the puzzles, executed by \tpc during $\gen{\mtdgm}$, are $O(n)$, which the costs reflect.
There is, however, a large constant gas cost inherent to the deployment of the contract.
Registering solutions is asymptotically constant. 
The minor variations seen in concrete gas cost are due to different parameters.

\begin{table}[!t]
\caption{ Cost comparison of initializing the contract, committing to the puzzle chain, and registering a solution across different networks.}
\label{tab:cost_comparison}
\setlength{\smalltabcolsep}{3pt}

\sisetup{
table-format=4.2,
round-mode = places,
round-precision = 2,
print-zero-integer=true,
table-text-alignment=right,
group-digits = none,
}

\begin{tabular}{
@{}
S[table-format=5.0, round-precision = 0]@{}
S[table-format=4.2]@{\hspace{\smalltabcolsep}}
S[table-format=4.2]@{\hspace{\smalltabcolsep}}
S[table-format=1.2]@{\hspace{\smalltabcolsep}}
S[table-format=3.2]@{\hspace{\smalltabcolsep}}
S[table-format=3.2]@{\hspace{\smalltabcolsep}}
S[table-format=1.2]@{\hspace{\smalltabcolsep}}
S[table-format=2.2]@{\hspace{\smalltabcolsep}}
S[table-format=1.2]@{\hspace{\smalltabcolsep}}
S[table-format=1.2]@{\hspace{\smalltabcolsep}}
} 
\toprule
&\multicolumn{9}{c}{Algorithms' Gas Cost (in USD$^*$\,)} \\ \cmidrule(l){2-10}

&\multicolumn{3}{c}{Etherium}
&\multicolumn{3}{c}{Arbitrum}
&\multicolumn{3}{c}{Polygon}\\
\cmidrule(rl){2-4}\cmidrule(rl){5-7}\cmidrule(l){8-10}

{\multirow{-4}{*} {\parbox{1cm}{\centering Number\\of\\puzzles}}}
& \text{init.}
& \text{com.}
& \text{reg.}
& \text{init.}
& \text{com.}
& \text{reg.}
& \text{init.}
& \text{com.}
& \text{reg.}
\\
\cmidrule(rl){1-1}\cmidrule(rl){2-4}\cmidrule(rl){5-7}\cmidrule(l){8-10}
 
 1
& 0.98
& 0.67
& 1.38
& 0.05
& 0.04
& 0.07
& 0.01
& 0.01
& 0.01
\\

 10
& 4.89
& 3.99
& 1.31
& 0.25
& 0.20
& 0.07
& 0.02
& 0.01
& 0.01
\\

 100
& 43.72
& 37.17
& 1.30
& 2.19
& 1.86
& 0.07
& 0.11
& 0.10
& 0.01
\\

 1000
& 433.44
& 369.10
& 1.30
& 21.67
& 18.46
& 0.07
& 1.08
& 0.92
& 0.01
\\

 10000
& 4331.40
& 3690.14
& 1.30
& 216.54
& 184.48
& 0.07
& 10.70
& 9.12
& 0.01
\\
\bottomrule

\multicolumn{10}{p{8cm}}{\footnotesize $^*$ Transaction USD price on 20 Aug 2024, 11:28 UTC using \url{www.cryptoneur.xyz/en/gas-fees-calculator}.}\\

\end{tabular}
\end{table}

\subsection{Further Discussion}
Delegation with fair payment incentives developing and making available specialized hardware for squaring-based puzzles and VDFs.
This levels the playing field between parties with different computational capabilities and addresses one of the fundamental issues of time-based cryptography: the heterogeneity of solvers which otherwise causes the 
solve time to vary substantially.
Figure~\ref{fig:bench} shows that commodity cloud and high-end consumer hardware are comparable in solving speed%
\footnote{
This may be explained by the focus on highly parallelizable tasks, especially for cloud computing; puzzle solving is inherently sequential.}%
.
In contrast, FPGA-based squaring requires only 4 cycles per operation~\cite{FPGACompRes}.
Therefore, speed-up by a factor of 15 can be achieved by delegating solving to a helper using the FPGA described in~\cite{FPGACompRes} compared to the fastest general-purpose CPU benchmarked (Apple M1).

As the overhead of \dgc is low, delegation and the choice of helper become economics decisions.
Even if \cl and \se have the capabilities to run \gc on their own, they can still consider \dgc if the speed-up offered by delegation provides enough value compared to the monetary cost of delegation.
Helpers with different capabilities and costs can be considered.
A slower helper with adequate performance for the needs of \se at a lower cost may be preferred over a costlier, faster helper.
In addition, if faster helpers are busy, slower helpers may provide the solution sooner overall, which an appropriate instantiation of \ced can capture.

Setting the value of $S$, the number of squarings per second the most capable server can perform, is currently a challenge as numbers are not readily available.
If a market for delegation of repeated squaring develops, puzzle generation can be calibrate to the strongest offerings on the market to ensure that puzzles cannot be solved before the intended time.
Additionally, by allowing \cl to extend the chain, this means new puzzles in a chain can be generated using updated information about the strongest servers.

\section{Applications of (E)D-TLP}\label{sec:applications}
We briefly discuss applications of our (E)D-TLP within the contexts of VDF and proofs of data storage

\subsection{Delegated VDF}\label{subsec:delegated-vdf}
Verifiable delay functions~(VDFs) let a prover provide a publicly verifiable proof that it has performed a predetermined number of sequential computations~\cite{BonehBBF18,BonehBF18}.
As VDF schemes have been built on TLPs~\cite{Wesolowski19,Pietrzak19a}, the ideas behind \dgm can also directly contribute to VDF scalability. 
A typical VDF involves algorithms:
\begin{enumerate*}[label=(\arabic*)]
    \item $\setup{}(\cdot)$ that returns system parameters,
    \item $\eval(\cdot)$ that returns a value $y$ and proof that the value has been constructed correctly, and
    \item $\verify{}(\cdot)$ that decides if a proof is valid.
\end{enumerate*}
By definition, $\eval(\cdot)$ demands considerable computational resources to compute a predefined number of sequential steps until it obtains the final result, similarly to TLP.

However, to date, the scalability of VDF schemes has been overlooked.
The literature has primarily concentrated on the specific case where the solver deals with only a single instance of VDF, ignoring the generic multi-instance VDF cases. 
Using a current VDF scheme, if a resource-constrained party (lacking a sufficient level of computation power) has to run multiple VDFs within a period, it must defer some instances of $\eval(\cdot)$ until others have completed. 
However, this leads to a significant delay in computing all VDFs instances. 
This becomes particularly problematic in a competitive environment, where different parties need to submit the output of VDF ($\eval(\cdot)$) instances to a smart contract on time to receive a reward, e.g., in~\cite{ThyagarajanGBKS21}. 
By adapting \dgm, VDFs can be delegated with fair payment for timely delivery, letting systems scale as the number of VDF instances increases beyond their immediate capabilities. 
\subsection{Scalable Proof of Storage}\label{subsec:proof-of-storage-for-multiple-files}
There are scenarios in which a client desires a server to acquire distinct ``random challenges'' at various points in time within a specified period, without the client's involvement during that timeframe.
Such challenges would enable the server to generate specific proofs, including, but not limited to, demonstrating the \textit{continuous availability of services}, such as proofs of data storage \cite{Abadi-C-TLP,Storage-Time,YuHYL21}.
In scenarios where a client outsources its file to a storage server and aims to ensure the file's integrity and accessibility at different intervals, these challenges become crucial. 
The idea involves the client computing random challenges, encoding them into puzzles, and transmitting them to the server. The server can then solve each puzzle, extract a subset of challenges, and employ them for the relevant proof scheme.

In cases where a (set of) client needs to simultaneously activate multiple instances of the scheme (e.g., invoking multiple instances of proofs of data storage scheme when the client has multiple independent files), thin clients with limited resources can not generate all the puzzles and the storage server with insufficient computational resources may struggle to solve the puzzles on time.

To surmount this challenge in a privacy-preserving manner, our ED-TLP can be employed to (i) enable the client to delegate the generation of puzzles to a more powerful server and (ii) allow the storage server to verifiably delegate the task of solving the puzzles to a resourceful helper.
In the multi-client case, the server can ask clients to combine their puzzles into a single chain to further reduce the overall workload.
\balance
\section{Related Work}\label{sec:related-work}
Since the introduction of the RSA-based TLP, various variants have been proposed.
For a comprehensive survey, we refer readers to a recent survey~\cite{medley2023sok}. 
\citeauthor{BonehN00}~\cite{BonehN00} and \citeauthor{GarayJ02}~\cite{GarayJ02} have proposed TLPs for the setting where a client can be malicious and needs to prove (in zero-knowledge) to a solver that the correct solution will be recovered after a certain time.
\citeauthor{BaumDDNO21}~\cite{BaumDDNO21} have developed a composable TLP that can be defined and proven in the universal composability framework. 
\citeauthor{MalavoltaT19}~\cite{MalavoltaT19} and \citeauthor{BrakerskiDGM19}~\cite{BrakerskiDGM19} proposed the notion of homomorphic TLPs, which let an arbitrary function run over puzzles before they are solved.
They use the RSA-based TLP and fully homomorphic encryption, which imposes high overheads.
To improve the above homomorphic TLPs, \citeauthor{SrinivasanLMNPT23}~\cite{SrinivasanLMNPT23} proposed a scheme that supports the unbounded batching of puzzles.

\citeauthor{ThyagarajanGBKS21}~\cite{ThyagarajanGBKS21} proposed a blockchain-based scheme that allows a party to delegate the computation of sequential squaring to a set of competing solvers.
To date, this is the only scheme that considers verifiably outsourcing sequential squaring for rewards.
In this scheme, a client posts the number of squarings required and its related public parameters to a smart contract, alongside a deposit.
The solvers are required to propose a solution before a certain time point. 
This scheme suffers from a security issue that allows colluding solvers to send incorrect results and invalid proofs, but still they are paid without being detected~\cite{Abadi2023-1347}.
Also, this scheme does not offer any solution for multi-puzzle settings.
It offers no formal definition for delegated time-lock puzzles (they only present a formal definition for delegated squaring) and supports only solver-side delegation.
Our (E)D-TLP  addresses these limitations.

\section{Conclusion}\label{sec:conclusion}

In this work, we introduced a scalable TLP tailored for real-world scenarios involving numerous puzzles, effectively addressing the computational constraints of clients and servers under heavy workloads. Specifically, we proposed the concept of Delegated Time-Lock Puzzle (\dgm) and developed a novel protocol, \dgc, to implement \dgm, significantly enhancing the scalability of TLPs. 
This protocol allows a server to offload its workload to a more capable helper with minimal overhead, ensuring timely delivery and fair payment. Delegating tasks to a helper with specialized hardware is particularly advantageous, as it enables significant speed-ups. In a multi-client scenario, \dgc optimizes efficiency by allowing clients to merge their puzzles into a single puzzle chain, equivalent in complexity to the longest puzzle.

We implemented \dgc, configured its concrete parameters, and evaluated its on-chain and off-chain costs. For the first time, we assessed the performance of state-of-the-art TLPs for handling a large number of puzzles. The results demonstrate that our scheme is both efficient and highly scalable. \dgc enables the delegation of 99\% of the client workload and 100\% of the server workload with minimal overhead, achieving smart contract gas costs as low as $0.2$ cents per puzzle. \dgc can enhance the scalability of systems reliant on delay-based cryptography.

\begin{acks}
Aydin Abadi was supported in part by REPHRAIN: The National Research Centre on Privacy, Harm Reduction and Adversarial Influence Online, under UKRI grant: EP/V011189/1.
Dan Ristea is funded by UK EPSRC grant EP/S022503/1 supporting the CDT in Cybersecurity at UCL\@.
Steven J.~Murdoch was supported by REPHRAIN.
\end{acks}

\bibliographystyle{ACM-Reference-Format}
\bibliography{bib/main-ref,bib/vdf-fpga,bib/repo}

\appendix
\section{Discussion on Network Delay Parameters}\label{sec::network-delay}

As discussed in~\cite{GarayKL15}, liveness states that
an honestly generated transaction will eventually be included more than $\kappa$ blocks deep in an honest party’s blockchain.
It is parameterized by wait time: $u$ and depth: $\kappa$.
They can be fixed by setting $\kappa$ as the minimum depth of a block considered as the blockchain’s state (i.e., a part of the blockchain that remains unchanged with a high probability, e.g., $\kappa \geq 6$) and $u$ the waiting time that the transaction gets $\kappa$ blocks deep.

As shown in~\cite{BadertscherMTZ17}, there is a slackness in honest parties' view of the blockchain.
In particular, there is no guarantee that at any given time, all honest miners have the same view of the blockchain or even the state.
But, there is an upper bound on the slackness, denoted by $\mathtt{WindowSize}$, after which all honest parties would have the same view on a certain part of the blockchain state.

This means when an honest party (e.g., the server) propagates its transaction (containing the proof) all honest parties will see it on their chain after at most: $\Upsilon = \mathtt{WindowSize} + u$ time period.

\section{Symmetric-key Encryption Scheme} \label{sec:SKE}

A symmetric-key encryption scheme consists of three algorithms:
\begin{enumerate*}[label=(\arabic*)]
\item $\skegen(1^\lambda)\rightarrow k$ is a probabilistic algorithm that outputs a symmetric key $k$.
\item $\skeenc(k,m)\rightarrow c$ takes as input $k$ and a message $m$ in some message space and outputs a ciphertext $c$.
\item $\skedec(k,c)\rightarrow m$ takes as input $k$ and a ciphertext $c$ and outputs a message $m$.
\end{enumerate*}.
The correctness requirement is that for all messages $m$ in the message space: $\Pr[\skedec(k, \skeenc(k, m))=m: \skegen(1^\lambda)\rightarrow k ]=1\;$.

The symmetric-key encryption scheme satisfies \emph{indistinguishability under chosen-plaintext attacks (IND-CPA)}, if for any PPT adversary \adv there is a negligible function $\mu(\cdot)$ for which \adv has no more than $\frac{1}{2}+\mu(\lambda)$ probability in winning the following game. 
The challenger generates a symmetric key by calling $\skegen(1^\lambda)\rightarrow k$.
The adversary \adv is given access to an encryption oracle $\skeenc(k,\cdot)$ and eventually sends to the challenger a pair of messages $m_0,m_1$ of equal length.
In turn, the challenger chooses a random bit $b$ and provides \adv with a ciphertext $\skeenc(k$, $m_b)\rightarrow c_b$.
Upon receiving $c_b$, \adv continues to have access to $\skeenc(k,\cdot)$ and wins if its guess $b'$ is equal to $b$.

\section{Notation}\label{sec:notation}

Table~\ref{table:notation-table} summarizes the key notations used in the paper.

\begin{table}[h]
\caption{\small{Notation Table}.}\label{table:notation-table}
\begin{tabularx}{\linewidth}{@{}rX}
\toprule
\textbf{Symbol}&  \textbf{Description} \\\midrule

$\lambda$ & Security parameter \\
\se & Server \\
\cl & Client \\
\tpc & Helper of client \\
\tps & Helper of server \\
\scc & Smart contract \\
$\mathcal{NP}$ & Nondeterministic polynomial time\\
\lang & Language in $\mathcal{NP}$\\
$N$ & RSA modulus \\
\tlp & Time-lock puzzle \\
\ctlp & Chained time-lock puzzle \\
\gm & Generic multi-instance time-lock puzzle \\
\dgm & Delegated time-lock puzzle \\
$\Delta$ & Period that a message must remain hidden \\
$S$ & Max.\ squarings the strongest solver performs per second \\
$T=S\cdot \Delta$ & Number of squarings needed to solve a puzzle \\
$z$& Number of puzzles \\
$[z]$& First $z$ natural numbers $[1,2,3, \ldots, z]$ \\
$j$ & Puzzle's index \\
$s$ or $s_j$ & Puzzle's solution \\
$p$ or $p_j$ & Puzzle \\
\puzzvec & A vector of puzzles $[p_1, \ldots, p_z]$ \\
$g_j$ & Public statement \\
\commvec & A vector of statements $[g_1,\ldots, g_z]$\\
$\hat{p}$ & $ \hat{p}\coloneqq(\puzzvec, \commvec)$\\
$m, m_j$ & Plaintext message\\
$r, r_j, d_j$ & Random value\\
\hash & Hash function\\
$\bar{\Delta}_j$ & Time interval \\
$\sum_{i=1}^j \bar{\Delta}_I$ & Period after which $j$-th solution is found in \gm/\dgm \\
$\pi_j$ & Proof of solution's correctness \\
\ced & Customized extra delay generating function \\
$ToC$ & Type of computational step \\
$aux, aux_{ID}$ &  Specification of a certain solver \\
$\Upsilon$ & Network delay \\
$m_j^*$ &  Encoded/encrypted message \\
$\messvec^*$ & A vector of encrypted messages $[m_1^*,\ldots, m_z^*]$ \\
$t_0$ & Time when puzzles are given to solver \\
$\coins_j$ & Coins paid to solver for finding $j$-th solution \\
\coins & $\coins_1 + \ldots + \coins_z$ \\
\coinsvec & A vector of coins $[\coins_1, \ldots, \coins_z]$ \\
$t_j$ &  Time when $j$-th solution is registered in \scc \\
$\Psi_j$ &  Extra time needed to find $j$-th solution \\
\extravec & A vector of extra time $[\Psi_1,\ldots, \Psi_z]$ \\
$\adr_I$ & Party $I$'s account addresses \\
$\mu$ & Negligible function \\
\skegen & Symm-key encryption's key gen.\ algo.\ \\
\skeenc & Symm-key encryption's encryption algo.\ \\
\skedec & Symm-key encryption's decryption algo.\ \\
\Com & Commitment scheme's commit algo.\ \\
\Ver & Commitment scheme's verify algo.\ \\
\bottomrule
\end{tabularx}
\end{table}

\section{Sequential and Iterated Functions}\label{sec::sequential-squaring}

\begin{definition} [$\Delta,\delta(\Delta))$-Sequential function]
For a function: $\delta(\Delta)$, time parameter: $\Delta$ and security parameter: $\lambda=O(\log(|X|))$, $f:X\rightarrow Y$ is a $(\Delta,\delta(\Delta))$-sequential function if the following conditions hold:
\begin{itemize}[leftmargin=.43cm]
\item[$\bullet$] There is an algorithm that for all $x\in X$ evaluates $f$ in parallel time $\Delta$, by using $\poly(\log(\Delta)$, $\lambda)$ processors.
\item[$\bullet$] For all adversaries $\adv$ which execute in parallel time strictly less than $\delta(\Delta)$ with $\poly(\Delta$, $\lambda)$ processors there is a negligible function $\mu(\cdot)$: 
\[\Pr\left[y_\adv\pick \adv(\lambda,x), x\pick X :  y_\adv=f(x) \right]\leq \mu(\lambda)\]
where $\delta(\Delta)=(1-\epsilon)\Delta$ and $\epsilon<1$.
\end{itemize}
\end{definition}

\begin{definition}[Iterated Sequential function] Let $\beta: X\rightarrow X$ be a $(\Delta,\delta(\Delta))$-sequential function.
A function $f: \mathbb{N}\times X\rightarrow X$ defined as $f(k,x)=\beta^{(k)}(x)=\overbrace{\beta\circ \beta\circ \ldots \circ \beta}^{k \text{\ \ Times}}$ is an iterated sequential function, with round function $\beta$, if for all $k=2^{o(\lambda)}$ function $h:X\rightarrow X$ defined by  $h(x)=f(k,x)$ is $(k\Delta,\delta(\Delta))$-sequential.

\end{definition}

The primary property of an iterated sequential function is that the iteration of the round function $\beta$ is the quickest way to evaluate the function.
Iterated squaring in a finite group of unknown order is widely believed to be a suitable candidate for an iterated sequential function.
Below, we restate its definition.

\begin{assumption}[Iterated Squaring]\label{assumption::SequentialSquaring} Let N be a strong RSA modulus, $r$ be a generator of $\mathbb{Z}_N$, $\Delta$ be a time parameter, and $T=\poly(\Delta,\lambda)$.
For any $\adv$, defined above, there is a negligible function $\mu(\cdot)$ such that:

\[\Pr\left[%
\begin{array}{l}%
r \pick \mathbb{Z}_N, \rb\pick \bin\\
\text{if } \rb=0,\ y \pick \mathbb{Z}_N\\
\text {else } y=r^{2^T}
\end{array}%
:%
\adv(N,r,y)\rightarrow \rb
\right]\leq \frac{1}{2}+\negl\]

\end{assumption}

\section{The original RSA-based TLP}\label{sec:RSA-based-TLP}
Below, we restate the original RSA-based time-lock puzzle proposed by~\citeauthor{Rivest:1996:TPT:888615}~\cite{Rivest:1996:TPT:888615}.

\begin{enumerate}
\item \brokenuline{$\setup{\tlp}(1^\lambda, \Delta, S)$}.
\begin{enumerate}

\item pick two large random prime numbers, $q_1$ and $q_2$.
Set $N=q_1\cdot q_2$ and compute Euler's totient function of $N$ as follows, $\phi(N)=(q_1-1)\cdot (q_2-1)$.
\item set $T=S\cdot \Delta$ the total number of squarings needed to decrypt an encrypted message $m$, where $S$ is the maximum number of squaring modulo $N$ per second that the (strongest) server can perform, and $\Delta$ is the period, in seconds, for which the message must remain private.

\item generate a key for the symmetric-key encryption: \\
$\skegen(1^\lambda)\rightarrow k$.

\item choose a uniformly random value $r$: $r\pick\mathbb{Z}^*_N$.
\item set $\ex=2^T\bmod \phi(N)$.
\item set $\pk\coloneqq(N,T,r)$ as the public key and $\sk\coloneqq(q_1,q_2,\ex,k)$ as the secret key.
\end{enumerate}

\item\brokenuline{$\gen{\tlp}(m,\pk,\sk)$}.
\label{Generate-Puzzle-}
\begin{enumerate}
\item encrypt the message under key $k$ using the symmetric-key encryption, as follows: $p_1= \skeenc(k,m)$.
\item encrypt the symmetric-key encryption key $k$, as follows: $p_2= k+r^\ex\bmod N$.
\item set $p\coloneqq(p_1,p_2)$ as puzzle and output the puzzle.
\end{enumerate}

\item\brokenuline{$\solve{\tlp}(pk, {p})$}.
    \begin{enumerate}
    \item find $b$, where $b=r^{2^T}\bmod N$, through repeated squaring of $r$ modulo $N$.
    \item decrypt the key's ciphertext, i.e., $k=p_2-b\bmod N$.
    \item decrypt the message's ciphertext, i.e., $m=\skedec(k, $ $p_1)$.
        Output the solution, $m$.
    \end{enumerate}
\end{enumerate}

\section{Extending a puzzle chain}\label{sec::extending-a-puzzle-chain}

\begin{enumerate}

\item\brokenuline{$\genext{\gm}(1^\lambda, \intervalvec', S', z, z', \pk, \sk)\rightarrow (\pk'$, $\sk')$}. \\ 
Involves \cl. Requires an existing set of keys $\pk, \sk$, the previous number of puzzles $z$ and the number of new puzzles $z'$.
    \begin{enumerate}
    \item $\forall j, z < j \leq z + z'$, set $T_j=S'\cdot\bar\Delta'_j$.
    \item Set $\timeexpvec'=[T_1,\ldots, T_z,T_{z+ 1}, \ldots, T_{z+z'}]$ containing values $\timeexpvec=[T_1,\ldots, T_z] \in \pk$.

    \item compute values $a_j$: 
    $\forall j, z < j \leq z + z': \ex_j=2^{T_j}\bmod \phi(N)$.
    This yields vector $\expvec'=[\ex_1,\ldots, $ $ \ex_z, \ex_{z+1},\ldots, $ $\ex_{z+z'}]$ containing values $\expvec=[\ex_1,\ldots, \ex_z] \in \sk$.

    \item choose $z'$ fixed size random generators: $r_j\pick \mathbb{Z}^*_N$, where $|r_j|=\omega_1$, $\forall j,\ z < j \leq z + z'$ and set $\randvec'=[r_2$, $\ldots$, $r_{z+1}$, $r_{z+2}$, $\ldots$, $r_{z+z'+1}]$ containing values $\randvec\in \sk$.
    Also, pick $z'$ random keys for the symmetric-key encryption and set $\keyvec'=[k_1$, $\ldots$, $k_z$, $k_{z+1}$, $\ldots$, $k_{z+z'}]$ containing $\keyvec \in \sk$.
    Choose $z'$ fixed size sufficiently large random values, where $|d_j|=\omega_2$ $\forall j,\ z < j \leq z + z'$ and set $\witnvec'=[d_1$, $\ldots$, $d_z$, $d_{z+1}$, $\ldots$, $d_{z+z'}]$ containing values $\witnvec\in \sk$.

    \item set public key $\pk'\coloneqq(\textit{aux}$, $N$, $\timeexpvec'$, $r_1$, $\omega_1$, $\omega_2)$ and secret key $\sk' \coloneqq (q_1, q_2$, $\expvec'$, $\keyvec'$, $\randvec'$, $\witnvec')$ where $\textit{aux}, r_1, \omega_1, \omega_2 \in \pk$ and $q_1, q_2 \in \sk$.
    Output $\pk'$ and $\sk'$.
    \end{enumerate}

\item\brokenuline{$\genext{\gm}(\messvec', \pk', \sk')\rightarrow \hat{p}'$}. Involves \cl. 
Encrypt the messages $\forall j, z < j \leq z + z'$ as it is done in $\gen{\gm}$ and output the new puzzle vector $\hat{p}'$. Send $\hat{p}'$ and $\pk'$ to \se.
The new $\commvec'$ is made public.
\end{enumerate}
\se proceeds to solve the newly generated puzzles as normal but only once all puzzles in the original chain have been solved.
As extending a message chain is equivalent to creating a longer chain initially, the security properties of the protocol are not affected.

\section{Squaring Capabilities of Various Devices}\label{Squaring-Capabilities}

Table~\ref{table:bench-} presents the results of experiments conducted on various devices to evaluate their squaring performance.

\begin{table}[!hb]
\vspace{-2mm}
\caption{Squaring capabilities of consumer and commodity cloud hardware for 1024-bit and 2048-bit RSA groups. 
}
\label{table:bench-}
\vspace{-3mm}
\begin{threeparttable}
\scalebox{0.75}{
\begin{tabular}{ccrrr} 

\toprule

&
& Freq
&\multicolumn{2}{c}{ Squarings per second } \\

\cmidrule{4-5}

\multirow{-2}{*}{ Machine}
&\multirow{-2}{*}{ CPU}
& (GHz)
&\multicolumn{1}{c}{ 1024-bit}
&\multicolumn{1}{c}{ 2048-bit} \\

\midrule

 Apple MacBook Pro
& Apple M1 Pro
& 3.23
& $2.260 \cdot 10^6$
& $0.845 \cdot 10^6$ \\

 Dell Latitude 5421
& Intel i7-11850H
& 2.50
& $2.025\cdot 10^6$
& $0.749\cdot 10^6$ \\

 University Cluster
& Intel Haswell
& 2.4
& $1.683 \cdot 10^6$
& $0.671 \cdot 10^6$ \\

 Azure F2s\_v2~\cite{AzureSizes,AzureNaming}
& Intel Xeon 8272CL
& 2.60
& $1.612\cdot 10^6$
& $0.628\cdot 10^6$ \\

 Azure DS2\_v2~\cite{AzureSizes,AzureNaming}
& Intel Xeon E5-2673 v3
& 2.40
& $1.201 \cdot 10^6$
& $0.480 \cdot 10^6$ \\

 Azure D2as\_v5~\cite{AzureSizes,AzureNaming}
& AMD EPYC 7763
& 2.45
& $1.935\cdot 10^6$
& $0.686\cdot 10^6$ \\

 Azure D2ls\_v5~\cite{AzureSizes,AzureNaming}
& Intel Xeon 8370C
& 2.80
& $1.509 \cdot 10^6$
& $0.573 \cdot 10^6$ \\

 Azure D2pls\_v5~\cite{AzureSizes,AzureNaming}
& Ampere Altra
& 3.0
& $0.824\cdot 10^6$
& $0.262\cdot 10^6$ \\

 Azure FX4mds~\cite{AzureSizes,AzureNaming}
& Intel Xeon 6246R CPU
& 3.40
& $1.911 \cdot 10^6$
& $0.736 \cdot 10^6$ \\

 \textit{Pearson FPGA\tnote{$*$}}
& \textit{AWS F1~\cite{Amazon_F1}}
& 0.1586
& \textit{$38.168 \cdot 10^6$}
& $N/A$ \\

\bottomrule
\end{tabular}
}
\begin{tablenotes}\footnotesize
\item[*] For comparison, we include the squaring capabilities of an FPGA-based solution\\ developed by Eric Pearson for the VDF~Alliance FPGA~Competition~\cite{FPGAComp}, as\\ reported in~\cite{FPGACompRes}. 
\end{tablenotes}
\end{threeparttable}
\vspace{-3mm}
\end{table}

\clearpage

\end{document}